\documentclass[a4paper,onecolumn,11pt,accepted=2024-09-19]{quantumarticle}
\pdfoutput=1
\usepackage[numbers,sort&compress]{natbib}
\usepackage{authblk}
\usepackage{tabu}
\usepackage{array}
\usepackage{amsmath}
\usepackage{amsthm}
\usepackage{amssymb}
\usepackage{algorithm}
\usepackage{braket}
\usepackage{subfig}
\usepackage{color}
\usepackage[dvipsnames]{xcolor}
\usepackage[english]{babel}
\usepackage{graphicx}
\usepackage{grffile}
\usepackage{wrapfig,epsfig}
\usepackage{epstopdf}
\usepackage{url}
\usepackage{color}
\usepackage{qcircuit}
\usepackage{epstopdf}
\usepackage{algpseudocode}
\usepackage[T1]{fontenc}
\usepackage{bbm}
\usepackage{comment}
\usepackage{dsfont}
\usepackage{mathtools}
\usepackage{enumitem}
\usepackage{pifont}
\usepackage{tikz}
\usetikzlibrary{arrows}
\usepackage[margin=1in]{geometry}
\graphicspath{{./figs/}}
\usepackage{hyperref}
\usepackage[nameinlink,capitalize]{cleveref}
\usepackage{textcomp}
\usepackage{multirow}

\DeclareMathOperator{\poly}{poly}

\renewcommand{\tt}{T_{\mathrm{total}}}

\definecolor{mygreen}{RGB}{80,180,0}
\definecolor{b2}{RGB}{51,153,255}

\newcommand{\nc}{\newcommand}

\nc{\nnl}{\nn \\ &}  
\nc{\fot}{\frac{1}{2}} 
\nc{\oo}[1]{\frac{1}{#1}} 
\newcommand{\ben}{\begin{enumerate}}
\newcommand{\een}{\end{enumerate}}
\nc{\mc}{\mathcal}
\graphicspath{{figures/}}
\nc{\onenorm}[1]{\L\| #1 \R\|_1} 

\nc{\Ra}{\Rightarrow}
\nc{\zo}{\{0,1\}}

\usepackage{authblk}


\newtheorem{theorem}{Theorem}
\newtheorem*{remark}{Remark}

\newtheorem{lemma}{Lemma}

\newcommand{\pr}{\text{Pr}}

\newcommand{\rmi}{\mathrm{i}}
\newcommand{\soff}{\sigma_{\mathrm{off}}}
\newcommand{\sH}{\sigma_{\mathrm{H}}}

\newcommand{\on}{\mathrm{on}}
\newcommand{\off}{\mathrm{off}}
\newcommand{\test}{\mathrm{test}}
\newcommand{\rmR}{\mathrm{R}}
\newcommand{\rmI}{\mathrm{I}}
\newcommand{\caT}{\mathcal{T}}
\newcommand{\dom}{\mathrm{dom}}
\newcommand{\caF}{\mathcal{F}}
\newcommand{\caD}{\mathcal{D}}

\newcommand{\caO}{\mathcal{O}}
\newcommand{\caK}{\mathcal{K}}

\newcommand{\bbS}{\mathbb{S}}

\newcommand{\fku}{\mathfrak{u}}
\newcommand{\res}{\mathrm{res}}
\newcommand{\vertiii}[1]{{\left\vert\kern-0.25ex\left\vert\kern-0.25ex\left\vert #1 
    \right\vert\kern-0.25ex\right\vert\kern-0.25ex\right\vert}}

\begin{document}

\title{Quantum Phase Estimation by Compressed Sensing}

\author{Changhao Yi}
\thanks{These authors contributed equally to this work.}
\affiliation{State Key Laboratory of Surface Physics, Department of Physics, and Center for Field Theory and Particle Physics, Fudan University, Shanghai, China}
\affiliation{Institute for Nanoelectronic Devices and Quantum Computing, Fudan University, Shanghai, China}
\affiliation{Shanghai Research Center for Quantum Sciences, Shanghai, China}
\author{Cunlu Zhou}
\thanks{These authors contributed equally to this work.}
\affiliation{Department of Computer Science \& Institut Quantique, Université de Sherbrooke, QC, Canada}
\affiliation{Center for Quantum Information and Control \& Department of Physics and Astronomy, University of New Mexico, NM, USA}
\author{Jun Takahashi}
\affiliation{Institute of Solid State Physics, University of Tokyo, Chiba, Japan}
\affiliation{Center for Quantum Information and Control \& Department of Physics and Astronomy, University of New Mexico, NM, USA}
\begin{abstract}
As a signal recovery algorithm, compressed sensing is particularly effective when the data has low complexity and samples are scarce, which aligns natually with the task of quantum phase estimation (QPE) on early fault-tolerant quantum computers. In this work, we present a new Heisenberg-limited, robust QPE algorithm based on compressed sensing, which requires only sparse and discrete sampling of times. Specifically, given multiple copies of a suitable initial state and queries to a specific unitary matrix, our algorithm can recover the phase with a total runtime of $\mathcal{O}(\epsilon^{-1}\text{poly}\log (\epsilon^{-1}))$, where $\epsilon$ is the desired accuracy. Additionally, the maximum runtime satisfies $T_{\max}\epsilon \ll \pi$, making it comparable to state-of-the-art algorithms. Furthermore, our result resolves the basis mismatch problem in certain cases by introducing an additional parameter to the traditional compressed sensing framework.
\end{abstract}

\maketitle

\section{Introduction}

Quantum phase estimation (QPE) \cite{kitaev1995quantum} is one of the most useful subroutines in quantum computing, playing a crucial role in many promising quantum applications \cite{shor1999polynomial,abrams1999quantum,mcardle2020quantum}. Given a unitary matrix $U$ and one of its eigenvectors $|\Phi\rangle$ with eigenphase $e^{i2\pi \theta}$, the task of QPE is to estimate phase $\theta$ within a specified accuracy. Assuming the unitary matrix $U$ represents the evolution operator under a Hamiltonian $H$, the task of QPE becomes equivalent to estimating a specific eigenenergy $E_0$ \cite{lin2022heisenberg,wang2022quantum}. Consequently, this subroutine has numerous applications in condensed matter physics, high-energy physics, and quantum chemistry. As a generalization, the problem of estimating multiple phases of $U$ has been referred to as the quantum eigenvalue estimation problem (QEEP) \cite{somma2019quantum,o2019quantum,lin2022heisenberg,zhang2022computing,wang2022quantum,dut22hei}.

While fully fault-tolerant quantum computers may still be years away, early fault-tolerant quantum computers with a limited number of logical qubits and gates are expected to emerge much sooner, capable of solving nontrivial tasks that demonstrate practical quantum advantages. Given the crucial role of QPE in many such tasks, it is imperative to design QPE algorithms specifically tailored for early fault-tolerant quantum computers. The standard textbook QPE algorithm \cite{nielsen2010quantum} does not require an exact eigenstate as the initial state and involves only one measurement, but it relies on a large number of ancilla qubits and controlled operations, making it experimentally demanding. Although Kitaev's original iterative QPE algorithm \cite{kitaev1995quantum} uses just one ancilla qubit and a single controlled operation (see Fig.~\ref{fig:hadamard}), it requires the initial state to be an exact eigenstate, which can be a challenging task itself. Therefore, neither approach is well-suited for early fault-tolerant quantum computers.

Most of the recent work \cite{lin2022heisenberg,PRXQuantum.4.020331,ni2023low} in QPE for early fault-tolerant quantum computers have focused on designing better protocols to improve various aspects of Kitaev's original QPE algorithm. More specifically, the following properties are desired when designing such algorithms:

\begin{itemize}
    \item The quantum circuit should be simple, using at most one ancilla qubit and one controlled operation.
    \item The initial state is not necessarily an exact eigenstate of $U$. 
    \item The total runtime achieves the Heisenberg limit, i.e., $\mathcal{O}(\epsilon^{-1}\poly\log(\epsilon^{-1}\delta^{-1}))$ for estimating the phase $\theta$ to accuracy $\epsilon$ with probability $1-\delta$.
    \item When the overlap of the initial state and the targeted eigenstate is large, the maximal runtime $T_{\max}$ (hence the maximum circuit depth) can be much smaller than $\pi/\epsilon$. 
\end{itemize}
In addition, the \emph{experimental complexity} should also be considered carefully for early fault-tolerant quantum computers: It is desirable to have a small number of time samples with a regular choice of values. For early fault-tolerant experiments, one may still need to prepare quantum circuits for each time $t$ by hand, and the total cost would be high if the number of different sample times is large. It is comparably easier to run the same quantum circuit multiple times, instead of running different quantum circuits a few times. Additionally, in quantum simulation algorithms \cite{georgescu2014quantum,childs2021theory} the time evolution by $U(t)$ is approximated by short-time evolutions: $U(t) \approx U_{\mathrm{sim}}(\Delta t)^L $. Thus, it is more convenient to sample from a discrete set of times $\mathcal{T} = \{n\Delta t, n\in \mathbb{N}\}$ rather than to sample from a continuous region. 
Most state-of-the-art algorithms sample times from a continuous region \cite{lin2022heisenberg,wang2022quantum,ding2023simultaneous} while only a few consider discrete sampling \cite{PhysRevA.108.062408, ding2024quantum}. 

For a Heisenberg-limited QPE algorithm with maximal runtime $T_{\max}$, if the size of the time samples needed is $\mathcal{O}(\mathrm{poly}\log T_{\max})$, the sampling is considered sparse in our paper. In this work, we design a non-adaptive algorithm that only requires discrete and sparse sampling of times without sacrificing performance. Informally, our main result can be stated as follows.


\begin{theorem}[Informal]
Suppose the target signal satisfies an approximately-on-grid assumption and its dominant part has frequency gap $N^{-1}$. Then our compressed sensing based QPE algorithm satisfies the following conditions:
    \begin{equation}
        |\mathcal{T}| = \mathcal{O}(\mathrm{poly}\log N), \quad T_{\max} = \mathcal{O}(N),\quad  T_{\mathrm{total}} = \mathcal{O}(N\mathrm{poly}\log N), \quad \epsilon = \mathcal{O}(\soff N^{-1}), 
    \end{equation}
where $|\mathcal{T}|$ is the number of samples, $T_{\max}$ is the maximum time for each time evolution, $T_{\mathrm{total}}$ is the total runtime, and $\epsilon$ is the desired accuracy, and $\soff$ quantifies the size of the off-grid component. The classical runtime of the algorithm is $\caO(\sigma_\off^{-1}N \poly\log N )$
\end{theorem}

A signal $y$ is approximately-on-grid if there exists a ``Fourier" matrix (with certain parameters) $F$ such that $y = Fx, x\in\mathbb{R}^{N}$. Note that under this assumption, we do not require $p_0 \ge \frac{1}{2}$, where $p_0$ is the overlap between the initial state and the ground state.  

The rest of the paper is organized as follows. We start with preliminaries about QEEP, sparse Fourier transformation and compressed sensing in Sec.~\ref{sec:preliminary}. We then introduce our QPE algorithm based on compressed sensing in Sec.~\ref{sec:on_grid} and prove several analytical results, including its Heisenberg-limit scaling. We also numerically test the performance of our algorithm and compare it to previous works in Sec.~\ref{sec:comparison}. In the final section, we summarize several open problems and potential future research directions in Sec.~\ref{sec:discussion}.

\section{Main idea}
\label{sec:preliminary}

\subsection{Setup}

The QEEP can be formulated as a sparse signal recovery problem. Given an initial state $|\Phi\rangle$ and a Hamiltonian with spectrum decomposition $H = \sum_{\ell=0}^{D-1} E_{\ell}P_{\ell}$, where $\{E_\ell\}_{\ell=0}^{D-1}$ are energy levels and $\{P_\ell = |\phi_\ell\rangle\langle\phi_\ell|\}_{\ell=0}^{D-1}$ are projectors onto the corresponding eigenstates, the time domain signal in QEEP can be written as 
\begin{equation}
    y^0(t) = \langle\Phi|e^{-\rmi Ht}|\Phi\rangle = \sum_{\ell=0}^{D-1}|\langle\Phi|\phi_{\ell}\rangle|^2 e^{-\rmi E_{\ell}t}.
\label{equ:time_domain_signal1}
\end{equation}
In QEEP we assume that $|\Phi\rangle$ has the following decomposition:
\begin{gather}
|\Phi\rangle = \sum_{\ell\in\mathcal{L}_{\mathrm{dom}}}\sqrt{p_{\ell}}|\phi_\ell\rangle + \sum_{\ell\in \mathcal{L}_{\mathrm{res}}}\sqrt{p_{\ell}}|\phi_\ell\rangle,\quad \sum_{\ell\in\mathcal{L}_{\mathrm{dom}}}p_\ell \approx 1,\quad |\mathcal{L}_{\mathrm{dom}}| \ll D,
\label{equ:Ldom}
\end{gather}
where $\mathcal{L}_{\mathrm{dom}}$ denotes the dominant component of the signal, and $\mathcal{L}_{\mathrm{res}}$ is the residue component. Under this assumption, we can regard $y^0(t)$ as a sparse signal. The formal definition of sparsity will be given in the main text. In particular, when $|\mathcal{L}_{\text{dom}}| = 1$, the task becomes QPE. For QPE, without loss of generality\footnote{In this work, we do not consider the hardness of state preparation. From the viewpoint of phase estimation, there is nothing special about the ground state energy compared to other eigenvalues as long as one can prepare an initial state that is close enough to the target eigenstate.}, we will mainly discuss the estimation of the ground energy $E_0$, i.e., the smallest eigenvalue of $H$.

The sparsity assumption applies to a wide range of situations. For instance, if we regard $|\Phi\rangle$ as the ground state of a perturbed Hamiltonian $H + V$, then the overlap $|\langle\Phi|\phi_\ell\rangle|$ tends to decay exponentially with the energy difference $|E_0 - E_\ell|$ (Because $|\langle\phi_\ell|V|\phi_0\rangle|$ decays exponentially with it. See \cite{arad2016connecting} for details). Therefore, $|\Phi\rangle$ almost has no overlap with excited states with high energies, and $\mathcal{L}_{\mathrm{dom}}$ only contains a small amount of energy levels.

The objective of a QEEP algorithm is to estimate $\mathcal{L}_{\text{dom}}$ within a specified accuracy $\epsilon$ using estimations of $y^0_t$ on a time set $\mathcal{T}$. An algorithm of this type can be separated into the quantum part and the classical post-processing part. Usually, the quantum part is a combination of Hamiltonian simulation \cite{georgescu2014quantum} and the Hadamard test (see Fig.~\ref{fig:hadamard}). Hamiltonian simulation algorithms are used to prepare the evolution operator $U(t)$. Longer evolution time requires more quantum gates, and the best-known circuit complexity for running $e^{-\rmi Ht}$ without ancilla qubit is almost linear in $\|Ht\|$ ($t$ can be negative) \cite{childs2019nearly,childs2021theory}.  The total runtime $T_{\text{total}}$ is thus
\begin{equation}
    T_{\text{total}} = \sum_{t\in \mathcal{T}} M_{\mathrm{H}} |t|.
\label{equ:totaltime}
\end{equation}
where $M_{\mathrm{H}}$ is the number of Hadamard tests required for each $y^0_t$. Usually, we set $M_\mathrm{H}$ to be in order $\mathcal{O}(\log(|\mathcal{T}|)\sH^{-2})$, so that $\{y^0_t\}_{t\in\mathcal{T}}$ can all be estimated within error $\sH$.   The total runtime $T_{\text{total}}$ reflects the total circuit depth for running the algorithm. If $T_{\text{total}} = \mathcal{O}(\epsilon^{-1}\mathrm{poly}(\epsilon^{-1}))$, we say that the algorithm satisfies the Heisenberg limit. Another important metric of complexity is the maximal runtime $T_{\max} = \max_{t\in\mathcal{T}}|t|$, which reflects the maximum circuit depth. Due to the difficulty in constructing large-size quantum circuits, the restriction on the maximal runtime is particularly important for early fault-tolerant quantum computers. 

The notations frequently used in the main text is summarized in Table~\ref{table:notations}.

\begin{table}[h!]
    \centering
    \caption{Notations}
    \begin{tabular}{ c c }
    \hline
    \hline
       Notation  &  Meaning \\
       \hline
        $y(t)$ & noisy time domain signal\\ 
        $y^0(t)$ & ideal time domain signal\\
        $z(t)$ & noise in time domain signal\\
        $x(k)$ & ideal frequency domain signal\\
        $\tau$ & unit time step\\
        $\mathcal{T}$ & sample of times divided by $\tau$\\
        $N$ & signal length, $T_{\max}/\tau$\\
        $\epsilon$ & accuracy on energy level\\
        $\sigma$ & noise tolerance parameter\\
        $\caF$ & set of frequencies\\
        $S$ & sparsity\\
    \hline
    \hline
    \end{tabular}
\label{table:notations}
\end{table}

\subsection{Previous work}

 The classical aspect of QPE and QEEP involves estimating frequencies from statistically sampled sparse signals, a process akin to the objectives of sparse Fourier transformation (SFT) algorithms \cite{hassanieh2012nearly}. Based on the data types, SFT algorithms can be classified into discrete setting algorithms and continuous setting algorithms. A discrete SFT algorithm performs discrete Fourier transformation, where both the time and the frequency of the signal are restricted on a discrete set:
\begin{equation}
    y_t = \sum_{k=0}^{N-1}x_k e^{-\rmi 2\pi kt/N},\quad \{y_t\}_{t=0}^{N-1} \xrightarrow{\mathrm{discrete \ SFT}} \{x_k\}_{k=0}^{N-1}.
\end{equation}
A continuous SFT algorithm \cite{liao2016music,song2023quartic,ding2024esprit} aims to accomplish a more general task:
\begin{equation}
    y(t) = \sum_{f\in\mathcal{F}}p_f e^{-\rmi 2\pi ft} \xrightarrow{\mathrm{continuous \ SFT}} x(k) = \sum_{f\in\mathcal{F}} p_f \delta(k-f),
\end{equation}
where $y(t)$ is the time domain signal, and $x(k)$ is the frequency domain signal. For QPE, we do not assume that the frequencies (energies) live in a discrete space. Thus, continuous SFT algorithms are more appropriate. In both setups, sparsity $S$ represents the number of distinct frequencies.

There are several aspects of evaluating the performance of an SFT algorithm. Its runtime complexity, sample complexity, and resolution are all important ingredients to consider. Here the runtime complexity refers to how long the algorithm takes on a classical computer, the sample complexity measures the number of time domain signal samples required in the algorithm, and the resolution quantifies the differences between the true frequencies and their estimates. For example, the Fast Fourier Transformation algorithm \cite{cochran1967fast} has runtime complexity $\mathcal{O}(N\log N)$ with sample complexity $\mathcal{O}(N)$. So far, the best runtime complexity is $\mathcal{O}(S\log^c(N)\log(N/S))$ with $c > 2$ \cite{gilbert2005improved}, and the most sample-efficient algorithm requires only $\mathcal{O}(S\log S \log N)$ samples \cite{indyk2014nearly}. In practical scenarios, we most likely have noisy data, necessitating the need for algorithmic robustness. Given the unique characteristics of our quantum setting, we prioritize the sample complexity, resolution, and robustness of an algorithm.

Several continuous SFT algorithms have been applied to QEEP. To the best of our knowledge, \cite{somma2019quantum} was the first work to address QEEP using Hadamard tests, treating it as a time-series analysis problem. Later, \cite{lin2022heisenberg} emphasized the importance of the Heisenberg-limited scaling. By applying Fourier-filter function techniques, it presented a Heisenberg-limited QPE algorithm for early fault-tolerant quantum computers. The algorithm was further improved by the other follow-up works \cite{zhang2022computing, wang2022quantum}, where the Gaussian derivative filter function was used to reduce the maximal runtime of the algorithm. In \cite{wang2022quantum},  $T_{\max}$ was reduced to a ``constant" depth, i.e., a quantity that only depends on the spectral gap, at the expense of increasing $T_{\mathrm{total}}$ from $O(\epsilon^{-1}\mathrm{poly}\log(\epsilon^{-1}))$ to $O(\epsilon^{-2}\mathrm{poly}\log(\epsilon^{-1}))$, which made the algorithm not Heisenberg-limited.

Two recent QPE algorithms \cite{PRXQuantum.4.020331, ni2023low}, inspired by Robust Phase Estimation (RPE) \cite{higgins2009demonstrating, kimmel2015robust, Bel20ach}, can also efficiently reduce the maximal runtime. These recent algorithms also improved the relation between $T_{\max}$, the initial overlap $p_0$, and the final accuracy $\epsilon$. When the overlap $p_0$ is large, \cite{PRXQuantum.4.020331} reduces the prefactor $\tau_c$ in the maximum runtime scaling $T_{\max} = \tau_c / \epsilon$ by using a subroutine called  the quantum complex exponential least squares (QCELS). In contrast to \cite{lin2022heisenberg} in which the prefactor $\tau_c$ is at least $\pi$, the prefactor in \cite{PRXQuantum.4.020331} can be arbitrarily close to $0$ as $p_0\to 1$. In \cite{ding2023simultaneous} and \cite{PhysRevA.108.062408}, the last two QPE algorithms have been extended to the QEEP setup.

 Another recent work \cite{ding2024quantum} proposed an efficient and versatile phase estimation algorithm named Quantum Multiple Eigenvalue Gaussian Filtered Search (QMEGS), which has most of the good properties mentioned above. Here we would like to emphasize its connection to a signal processing algorithm named Orthogonal Matching Pursuit (OMP) \cite{cai2011orthogonal}. OMP is a greedy algorithm that searches for the dominant frequencies of a signal by maximizing the overlaps step by step. QMEGS can be regarded as an OMP algorithm with a modified time sampling procedure to reduce the maximal and total runtime. The OMP algorithm has a strong connection with compressed sensing and can be potentially combined with our algorithm.

\subsection{QPE by compressed sensing}
\label{sec:mainidea}

Our main contribution is a simple and robust classical post-processing algorithm for QPE based on compressed sensing \cite{candes2006near,candes2008restricted,candes2006robust}. Our algorithm only requires sparse sampling of times from a discrete set.

Compressed sensing is a prominent signal-processing algorithm with wide applications in various domains such as time-frequency analysis, image processing, and quantum state and process tomography \cite{gross2010quantum,magesan2013compressing,Smith2013,magesan2013compressing,Kalev2015}. It aims to solve special types of underdetermined linear inverse problems, i.e.,  given $y\in \mathbb{R}^{M}$ and $A\in \mathbb{R}^{M\times N}$ with $M \ll N$, finding the unique sparse solution to $Ax = y, x\in\mathbb{R}^{N}$. Certainly, the solution is not unique without further restrictions. If we assume $x$ is $S$-sparse and $A$ satisfies the restricted isometry property (RIP) \cite{candes2008restricted} over $\caO(S)$-sparse signals, then $x$ can be uniquely recovered by solving a linear programming problem:
\begin{equation}
    \min_{x\in\mathbb{R}^N}\|x\|_1,\quad \mathrm{s.t.}\quad Ax = y.
\end{equation}
If we set $x$ as the frequency domain signal, $y$ as the signal on the time samples, and $A$ as the partial Fourier transformation operator, then this compressed sensing subroutine can be used for discrete SFT. It has been proved that with $\mathcal{O}(S\poly\log N)$ number of samples, one can successfully recover the frequency domain signal $x$ with high probability \cite{candes2006near}. For noisy situations, the signal can still be recovered by solving the following quadratic programming problem \cite{candes2008restricted}:
\begin{equation}
    \min_{x\in\mathbb{R}^N}\|x\|_1,\quad \mathrm{s.t.}\quad \|Ax - y\|_2 \le b.
\end{equation}
The small number of required samples and the robustness against noisy sampling make compressed sensing an appealing post-processing algorithm.

Unfortunately, there is a significant drawback of compressed sensing: it only works for discrete SFT, not for continuous SFT. In other words, frequencies are assumed to be on a grid in the sense that $f = n/N, n\in[N]$. The on-grid assumption is unnatural for many signals in practice. The gap between the continuous world and the discrete model is formally termed as \textit{basis mismatch} in signal analysis. Although off-grid compressed sensing algorithms have been proposed for solving the basis mismatch problem \cite{tang2013compressed}, the performance in our numerical test is not ideal. We show that with a slight modification, the vanilla compressed sensing algorithm can be used for special types of continuous SFT tasks. In other words, our algorithm partially solves the basis mismatch problem.

An overview of our algorithm is described as follows. For signal vectors with size $N$, when the frequencies are all nearly on-grid ($f \approx n/N, n\in\mathbb{Z}$) and the noise for each sample is bounded by a constant, the convex relaxation algorithm can recover the frequencies with only $\mathcal{O}(\poly\log N)$ samples, which satisfies the Heisenberg limit. With no prior knowledge about $f$ (i.e., $f$ could be off-grid), we introduce a grid shift parameter $\nu$ such that after shifting the signal by $e^{-\rmi 2\pi ft} \to e^{-\rmi 2\pi(f - \nu/N)t}$, the dominant frequencies of the new signal become nearly on-grid. This step requires an assumption on the signal, but we will show that a wide range of signals satisfy such an assumption. For each trial of $\nu$, we run the compressed sensing subroutine on the data set $\{y_t\}_{t\in\mathcal{T}}$ to obtain a trial solution $s_\nu$. The optimal $\nu$ is the one with the smallest $\|s_\nu\|_1$. By searching for the optimal grid-shift parameter in a finite set $\mathcal{V}$, the accuracy of the dominant frequencies is $\mathcal{O}(\soff/N)$, where $\soff$ is the maximal entry of the minimal off-grid component. This quantity is related to the noise, the frequency gap, and the residual part of the signal. In terms of the maximum runtime $T_{\max}$, since the samples of the compressed sensing algorithm are integers in $[1,N]$, $T_{\max}$ scales linearly in $N$, and $T_{\text{total}}$ is $\mathcal{O}(N\poly\log N)$.

\section{Main results}
\label{sec:on_grid}

\subsection{Algorithm}

\begin{figure}[h!]
    \centering
    \includegraphics[width = 8cm]{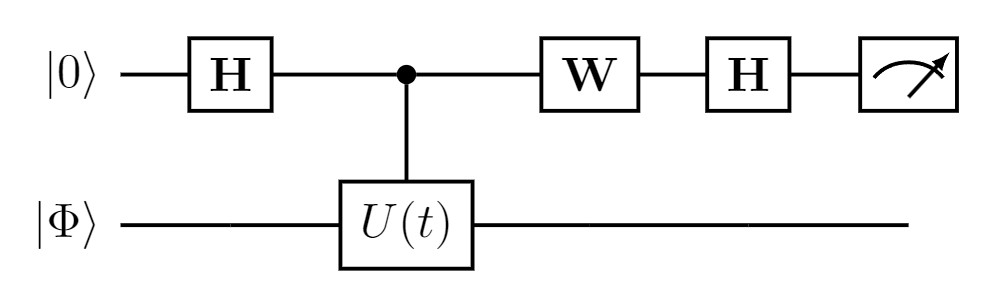}
    \caption{The one-ancilla quantum circuit used in Kitaev-type QPE algorithms. The final measurement is done in the $Z$ basis. In terms of the measurement outcome, we regard the $|0\rangle$ state as obtaining value $+1$, and the $|1\rangle$ state as obtaining value $-1$. $\mathbf{H}$ is the Hadamard gate; $\mathbf{W}$ has two choices: when $\mathbf{W} = I$, the measurement outcome is $\pm 1$ with probability $(1 \pm \text{Re}(\langle\Phi|U(t)|\Phi\rangle))/2$ respectively. When $\mathbf{W} = \bbS^{\dagger}$, the complex conjugation of the phase gate, the measurement outcome is $\pm 1$ with probability $(1 \pm \text{Im}(\langle\Phi|U(t)|\Phi\rangle))/2$ instead. After averaging over many test outcomes, we obtain an estimate of the true signal $\langle\Phi|U(t)|\Phi\rangle$.}\label{fig:hadamard}
\end{figure}

In this subsection, we present an overview of our algorithm for QPE using compressed sensing. The quantum part of the algorithm is represented as follows. 

\begin{algorithm}
\caption{Signal estimation by Hadamard test}\label{euclid}
\begin{algorithmic}[1]
 \Require{Set of sampled integers $\mathcal{T}$, unit time step $\tau$,  Hamiltonian $H$, initial state $|\Phi\rangle$, error tolerance parameter $\sH$, failure probability $\delta$.}
 \Ensure{$\{y(n\tau), n\in\mathcal{T}\}$.}
    \For{$n\in \mathcal{T}$}
        \State Prepare the initial state $|\Phi\rangle$ and unitary operator $e^{-\rmi Hn\tau}$;
        \State Perform Hadamard tests for $\mathcal{O}(\log(|\mathcal{T}|/\delta)\sH^{-2})$ times;
        \State Compute the average value of the test outcomes as $y(n\tau)$.
    \EndFor
\end{algorithmic}
\label{algorithm:hadamard_test}
\end{algorithm}

For each time $t\in\mathcal{T}$, $y^0(t) = \langle\Phi|e^{-\rmi Ht}|\Phi\rangle$ can be estimated from averaging over the Hadamard tests. More precisely, by choosing $\mathbf{W} = I$ in Fig.~\ref{fig:hadamard}, the measurement outcome of the ancilla qubit is a random variable $h_x(t)$ such that 
\begin{equation}
    \mathrm{Pr}[h_x(t) = +1] = \frac{1}{2}[1 + \text{Re}(y^0(t))],\quad \mathrm{Pr}[h_x(t) = -1] = \frac{1}{2}[1 - \text{Re}(y^0(t))].
\end{equation}
Similarly, when $\mathbf{W} = \bbS^{\dagger}$ ($\bbS$ is the phase gate), the measurement outcome is another random variable $h_y(t)$ such that
\begin{equation}
    \mathrm{Pr}[h_y(t) = +1] = \frac{1}{2}[1 + \text{Im}(y^0(t))],\quad \mathrm{Pr}[h_y(t) = -1] = \frac{1}{2}[1 - \text{Im}(y^0(t))].
\end{equation}
The summation of the two gives us the estimate of $y^0(t)$:
\begin{equation}
   \mathbb{E}[h_x(t) + \rmi h_y(t)] = y^0(t).
\end{equation}
After sampling the random variables $h_x(t),h_y(t)$ for $M_\mathrm{H}$ times, we obtain a noisy signal:
\begin{equation}
    y(t) = \overline{h_x(t) + \rmi h_y(t)} = y^0(t) + z(t).
\end{equation}
 Here the overline represents the average over samples, and the noise $z(t)$ originates from the statistical uncertainty of the Hadamard tests. Hoeffding's inequality ensures that with probability $1-\delta'$, we have
 \begin{equation}
     |z(t)|   =   \mathcal{O}\left(\sqrt{\frac{1}{M_\mathrm{H}}\log\frac{1}{\delta'}}\right).
 \end{equation}
In the rest of the paper, the meanings of $z(t)$ are not identical, but they always represent the noisy part of the signal. Let $\sH$ be the noise tolerance parameter. To guarantee $|z(t)| < \sH, \forall t\in\mathcal{T}$, with probability at least $1 - \delta$, we require $\delta = \mathcal{O}(\delta'|\mathcal{T}|^{-1})$ so that $M_{\mathrm{H}} = \Omega(\log(|\mathcal{T}|/\delta)/\sH^2)$. For a more rigorous proof, see Appendix A of \cite{PRXQuantum.4.020331}. The total runtime is thus
\begin{equation}
    T_{\text{total}} = \sum_{t\in \mathcal{T}} M_\mathrm{H}|t| = \mathcal{O}\left(\log(|\mathcal{T}|\delta^{-1})\sH^{-2} \sum_{t\in|\mathcal{T}|}|t|\right).
\end{equation}
The quantum part of the algorithm is shown in Algorithm~\ref{algorithm:hadamard_test}.

\begin{algorithm}[t]
\caption{Quantum phase estimation by compressed sensing}\label{euclid}
\begin{algorithmic}[1]
 \Require{Signal length $N$, signal sparsity $S$, sampling ratio $r$, unit time step $\tau$, Hamiltonian $H$, initial state $|\Phi\rangle$, noise tolerance parameter $\sH,\sigma_\test,\sigma$, failure probability $\delta$, size of the trial set $J$, amplitude lower bound $p_{\min}$.}
 \Ensure{$\nu_\ast = \nu_{j^\ast},\ s_{\nu_\ast} = s_{\nu_{j^\ast}},\ E^{\ast} = 2\pi (\min \mathcal{K} + \nu_{j^\ast})/N$.}
    \State Sample integers from $[N]$ with sampling ratio $r$. Denote the set of the samples by $\mathcal{T}$.
    \State Apply Algorithm~\ref{algorithm:hadamard_test} with input $(\mathcal{T},\tau,H,|\Phi\rangle,\sH,\delta)$. Denote the output by $\{y_n\}_{n\in\mathcal{T}_1}$.
    \For{$j = 0,1,\cdots,J-1$}
        \State Set $\nu_j = -1/2 + j/J$.
        \State Solve 
        \begin{equation*}
            \min_{s\in\mathbb{R}^N} \|s\|_1,\quad \mathrm{s.t.}\quad \|F_{\nu_j,\mathcal{T}}s - y_{\mathcal{T}}\|_2 \le \sqrt{|\mathcal{T}|}\sigma
        \end{equation*}
        to obtain $s_{\nu_j}$. If there is no feasible solution, set $s_{\nu_j} = (1,1,\cdots,1)$.
        \State Record $(\nu_j,s_{\nu_j})$ as the $j$-th solution.
    \EndFor
    \State Sample integers from $[N]$ with sampling ratio $r$. Denote the set of samples by $\mathcal{T}_2$. 
    \State Apply Algorithm~\ref{algorithm:hadamard_test} with input $(\mathcal{T}_2, \tau, H,|\Phi\rangle, \sH, \delta)$. Denote the output by $\{y'_n\}_{n\in\mathcal{T}_2}$.
    \For{$j = 0,1,\cdots,J-1$}
        \State Apply Algorithm~\ref{algorithm:test} with input $(\nu_j,s_{\nu_j},\{y'_n\}_{n\in\mathcal{T}_2}, \sigma_\test)$. Denote the output by $\mathrm{o}_j$.
        \If{$\mathrm{o}_j=0$}\State $\ell_j = N+1$. \Else \State $\ell_j = \|s_{\nu_j}\|_1$.
        \EndIf
    \EndFor
    \State Let $j^\ast = \arg\min \ell_j$
    \State Find all entries of $s_{\nu_{j^\ast}}$ satisfying $s_{\nu_{j^\ast}} \ge p_{\min}$. Denote the set of indices as $\mathcal{K}$. 
\end{algorithmic}
\label{algorithm:single_eigenvalue}
\end{algorithm}

\begin{algorithm}
\caption{Test of another sampling}\label{euclid}
    \begin{algorithmic}[1]
     \Require{Parameter $\nu$, frequency-domain signal $s_\nu$, sampled data $\{y_m\}_{m\in\mathcal{T}_2}$, mean squared error threshold 
 $\sigma_\test$.}
     \Ensure{0 if the data fails the test; 1 if the data passes the test.}
        \State Compute the total empirical error with respect to the new set
        \begin{equation*}\mathcal{E} = \sum_{m\in\mathcal{T}_2}|(F_{\nu} s_\nu)_m - y_m|^2.
        \end{equation*}
        \If{$\mathcal{E} \ge |\mathcal{T}_2|\sigma_\test^2$}
            \State Return 0.
        \Else \quad Return 1.
        \EndIf 
    \end{algorithmic}
\label{algorithm:test}
\end{algorithm}

For the classical post-processing part, the goal is to recover the dominant frequencies of $y^0_t$ with the noisy samples $\{y(t)\}_{ t\in\mathcal{T}}$. To write the QEEP in Eq.~(\ref{equ:time_domain_signal1}) in the form of a compressed sensing problem, we first reformulate the problem on a ``grid''. Introduce a unit time step $\tau$ such that
\begin{equation}
    y^0_n = \sum_{f\in\mathcal{F}} p_f e^{-\rmi 2\pi f n},\quad \mathcal{F} = \left\{\frac{E_{\ell}\tau}{2\pi} : \ \ell\in \mathcal{L}_{\text{dom}}\right\}.
\end{equation}
The dominant eigenenergy set $\mathcal{L}_{\text{dom}}$ defined in Eq.~(\ref{equ:Ldom}) determines the frequency support $\mathcal{F}$. To keep the order of energy levels unchanged, we require that $E_{\ell}\tau \in [0,2\pi), \ \forall \ell$.
This condition can always be satisfied by adding a constant to the Hamiltonian $H$ and choosing $\tau$ properly. The actual data to be processed is thus $\{y_n = y^0_n + z_n\}_{n\in\mathcal{T}}$. Recall that the choice of $M_H$ guarantees that $|z_n| \le \sH\ \forall n\in\mathcal{T}$, with high probability.

As discussed in Sec.~\ref{sec:mainidea}, because we cannot always assume $f\approx n/N,\ \forall f\in\mathcal{F}$, the regular compressed sensing algorithm is not guaranteed to work. Our algorithm significantly relaxes the assumption by introducing a grid-shift parameter $\nu$. As a simple instance, suppose the frequency support $\mathcal{F}$ satisfies
\begin{equation}
   f = \frac{n + \nu}{N}, \quad \ n \in [N],\ \forall f\in\mathcal{F}.
\label{equ:on-grid-condition}
\end{equation}
Then $y^0_n$ can be transformed into an on-grid signal in a new basis as $y^0_n = F_{\nu} x, x \in\mathbb{R}^N$, where $F_{\nu}$ is the shifted Fourier transformation:
\begin{equation}
    (F_{\nu})_{nk} \equiv e^{-\rmi 2\pi(k+\nu)n/N}\quad n,k\in[N].
\end{equation}
The signal can then be recovered by solving
\begin{equation}
    \min_{s\in\mathbb{R}^N} \|s\|_1,\quad \mathrm{s.t.}\quad \left\|\mathcal{P}_\mathcal{T}\left(F_{\nu}s - y\right)\right\|_2 \le \sqrt{|\mathcal{T}|}\sigma,
\end{equation}
where $\mathcal{P}_{\mathcal{T}}$ represents the projector
\begin{equation}
    \left(\mathcal{P}_{\mathcal{T}}\right)_{ij} = \delta_{ij}1_{i\in\mathcal{T}}.
\end{equation}
Define $F_{\nu,\mathcal{T}} \equiv \mathcal{P}_{\mathcal{T}}F_\nu$ and $y_\mathcal{T} \equiv \mathcal{P}_\mathcal{T}y$ for later use. One can argue that a general signal does not satisfy the condition in Eq.~(\ref{equ:on-grid-condition}), even approximately. We will address this issue in the next section and validate the universality of our algorithm with numerical evidence. 

In the case that the condition in Eq.~(\ref{equ:on-grid-condition}) is indeed satisfied, we can find $\nu$ by a brute-force search. Let $\mathcal{V}$ be a trial set of the grid shift parameters. For each $\nu\in\mathcal{V}$\footnote{Throughout this paper, we require that $|\nu| \le 1/2$.}, let $s_{\nu}$ be the solution of the compressed sensing subroutine and $s_{\nu_\ast}$ be the final solution with the smallest 1-norm: 
\begin{equation}
    \nu_\ast \equiv \arg\min_{\nu\in\mathcal{V}}\|s_{\nu}\|_1.
\end{equation}

To bound the algorithmic error properly, we need a rough bound for $|\nu_\ast - \mathfrak{u}|$ first, where $\mathfrak{u}$ is the optimal parameter introduced in the next section. This can be guaranteed by the test of another sampling (Algorithm~\ref{algorithm:test}). The intuition is as follows. The output of a compressed sensing subroutine always matches the true signal on $\mathcal{T}$. If the solution is ``good'' in the sense that $|\nu_\ast - \mathfrak{u}|$ is smaller than a critical value, then the recovered signal should be close to the true signal on another random sample set $\mathcal{T}_2$. If the solution is ``bad'', then the difference between the two signals will be very large on $\mathcal{T}_2$. 
The classical post-processing part of the algorithm is stated in Algorithm~\ref{algorithm:single_eigenvalue} with Algorithm~\ref{algorithm:test} as a subroutine.

\subsection{Statement of the main theorem}

The error bound and runtime of our algorithm are given in this subsection. Suppose the original signal is
\begin{equation} \label{equ:signal}
    y^0_t = \sum_{f\in\mathcal{F}}p_f e^{-\rmi 2\pi ft},\quad f = \frac{n_f + \nu_f}{N},
\end{equation}
where $n_f$ ($\nu_f$) is the integer (decimal) part of frequency $f$. Given a Fourier matrix $F_{\nu}$, $y^0$ can be uniquely decomposed as $y^0 = F_\nu(x_{\mathrm{re}} + \rmi x_{\mathrm{im}}),\ x_{\mathrm{re}},  x_{\mathrm{im}}\in \mathbb{R}^N$, and $F_\nu x_{\mathrm{re}}$ ($\rmi F_\nu x_{\mathrm{im}}$) is termed as the on-grid (off-grid) component of $y^0$ with respect to $F_{\nu}$. Define
\begin{equation}\label{equ:opt_nu}
    \mathfrak{u} \equiv \arg\min_\nu \|F_\nu x_{\mathrm{im}}\|_2,
\end{equation}
and denote the corresponding grid decomposition as
\begin{equation} \label{equ:xonxoff}
    y^0 = F_{\mathfrak{u}} (x_{\on} + \rmi x_{\mathrm{off}})\quad x_{\on}, x_{\off} \in \mathbb{R}^N,
\end{equation}
which we call the optimal grid decomposition of $y^0$. 

Let $\sigma_{\off} \equiv \|F_{\fku}x_{\off}\|_{\infty}$. We assume that the signal $y^0_t$ satisfies the following conditions:
\begin{equation}\label{equ:assumption}
\begin{aligned}
    \caF = \caF_\dom \bigsqcup \caF_\res,\quad & \min_{f\in\caF_\dom}p_f \ge p_{\min},\quad \sum_{f\in\caF_\res}p_f = \caO(\sigma_\off),\\
    |\caF_{\dom}| = S, & \quad \min_{f_1\neq f_2\in\caF_{\dom}}|f_1 - f_2| > \frac{1}{N}.
\end{aligned}
\end{equation}

In contrast to other QPE algorithms \cite{PRXQuantum.4.020331,ni2023low} where $p_0 > \frac{1}{2}$ is required, our algorithm outputs the dominant on-grid approximation of the signal ($y_t \approx \sum_n p_n e^{-\rmi 2\pi nt/N}$), instead of the dominant single-frequency approximation of the signal ($y_t \approx p e^{-\rmi 2\pi ft}, f\in [0,1)$). Hence, even if $p_0 < \frac{1}{2}$, as long as the dominant part is large enough, the ground-state energy can still be well estimated.
 
Even signals with no frequency gap can be put into this framework. Indeed, when $|f_1 - f_2| \ll  N^{-1}$, the two frequencies can be replaced by a single frequency $f_3$, and the tiny frequency gap is absorbed into the off-grid component. If $f_3$ can be approximated with high accuracy, it can be used as an estimate for $f_1$ because $|f_1 - f_3| \ll N^{-1}$.

In Appendix~\ref{app:proof_of_lemma1}, we prove the following lemma.
\begin{lemma}\label{lemma:off_grid}
    Given a length-$N$ signal $y^0$ with optimal grid decomposition $y^0 = F_{\mathfrak{u}}(x_{\on} + \rmi x_{\off}), \soff \equiv \|F_{\mathfrak{u}} x_{\off}\|_{\infty}$. If $y^0$ satisfies the assumptions in Eq.~(\ref{equ:assumption}), then for all $f\in\caF_{\dom}$, $|p_f (\nu_f - \fku)| = \caO(\sigma_\off)$.
\end{lemma}

From Lemma \ref{lemma:off_grid}, we know that if $\soff$ is small, then all $f$ with large enough $p_f$ are approximately on-grid simultaneously. 

Define
\begin{equation}
    \caD \equiv \left\{n: \exists f\in\caF_\dom, |Nf - n| \le \frac{1}{2}\right\}.
\end{equation}
Note that $|\caD| = |\caF_\dom| = S$. In the next lemma, we show that the assumptions in Eq.~(\ref{equ:assumption}) ensure that $\caD$ is the dominant part of $x_{\on}$ as well. The proof is presented in Appendix~\ref{app:proofofrecoverlemma}.
\begin{lemma}\label{lemma:recover}
    Suppose that the length-$N$ ($N\ge 100$) signal satisfies Eq.~(\ref{equ:assumption}). Then for any $f\in\caF_\dom, |Nf - n| \le 1/2$, we have
    \begin{equation}
        x_{\on,n} = p_f - \caO(S\sigma_\off),
    \end{equation}
    and for all $n\not\in \caD$, we have 
    \begin{equation}
        x_{\on,n} = \caO(S\sigma_\off). 
    \end{equation}
\end{lemma}

If $\caD$ can be recovered using Algorithm~\ref{algorithm:single_eigenvalue}, and $\fku$ is estimated by $\nu_{\ast}$, then all $f = (n+\fku)/N\in\caF_\dom$ are estimated by $(n+\nu_\ast)/N$, and their error bounds are all $\mathcal{O}(|\nu_\ast - \fku|/N)$. The formal error bound of our protocol is guaranteed by the following arguments.
\begin{itemize}
\item Let $\nu_1$ be the parameter in the trial set $\mathcal{V}$ that is closest to $\mathfrak{u}$, hence $|\nu_1 - \mathfrak{u}| = \mathcal{O}(|\mathcal{V}|^{-1})$. If $|\mathcal{V}|^{-1}$ is small enough, then all dominant frequencies of $x_{\mathrm{on}}$ are preserved in $s_{\nu_1}$. 
 
\item Recall that the solution with the smallest 1-norm is $s_{\nu_\ast}$. Using the test of another sampling, we can obtain a rough bound for $|\nu_\ast - \mathfrak{u}|$. Then we can prove that $s_{\nu_\ast}$ is close enough to $s_{\nu_1}$, so that all the dominant frequencies of $s_{\nu_1}$ are preserved in $s_{\nu_\ast}$.
 
\item Because all dominant frequencies of $x_{\mathrm{on}}$ are preserved in $s_{\nu_\ast}$, the accuracy of the final result is $\mathcal{O}(|\nu^\ast - \fku|/N)$.

\end{itemize}

Combining these arguments together, we obtain the following theorem regarding the accuracy and complexity of our algorithm. The formal proof can be found in Appendix~\ref{app:proof_of_main_theorem}.

\begin{theorem}

Suppose $y^0$ is a length-$N$ ($N \ge 100$) signal with optimal grid decomposition $y^0 = F_{\mathfrak{u}}(x_{\on}  + \rmi x_\off), \soff \equiv \|F_{\fku}x_{\off}\|_{\infty}$, and satisfies Eq.~(\ref{equ:assumption}). If $\|x_{\on}\|_2 \gg \|x_{\on}\|_1/\sqrt{N}$,  $E_0\tau/2\pi\in \caF_\dom$, and the parameters in Algorithm~\ref{algorithm:single_eigenvalue} satisfy 
    \begin{equation}
        4\pi \sqrt{S} < \sqrt{3}\log N,\quad 
        r = \mathcal{O}(N^{-1} S \log^2(S)\poly\log N),\quad\delta^{-1} = \mathcal{O}(\poly\log N),
    \end{equation}
    and $\exists \sigma_0 > 0$ such that
    \begin{equation} \label{equ:sigmas}
    \begin{aligned}
    p_{\min} \gg S\sigma_0,\quad 
        & J \ge \left\lceil \frac{\log N}{\sqrt{S}\sigma_0}\right\rceil,\quad\sigma_0 \ge \soff, \quad \sigma_0 \ge \sH,\\ 
        \sigma = \caO(\sigma_0) & \ge 2\sigma_0\left(1 - \frac{4\pi\sqrt{S}}{\log N \sqrt{3}}\right)^{-1},
    \end{aligned}
    \end{equation}
    then there exists $\sigma_\test = \mathcal{O}\left(\sigma_0\right)$ such that with probability at least $1 - 1/\mathrm{poly}(N)$, the output of Algorithm~\ref{algorithm:single_eigenvalue} $(\nu_\ast,s_{\nu_\ast}, E^\ast)$ satisfies 
    \begin{align}
        |\nu_\ast - \mathfrak{u}| = \mathcal{O}\left(\sigma_0\right),\quad \|s_{\nu_\ast} - x_{\mathrm{on}}\|_2 = \mathcal{O}\left(\sigma_0\right),\quad |E^\ast - E_0| = \mathcal{O}\left(\frac{\sigma_0}{Np_{\min}}\right).
    \end{align}
\label{theorem:main}
\end{theorem}
Note that the parameter $\sigma_0$ is essentially determined by $\sigma_\off$. If $\soff = 0$, meaning the $y^0$ is on-grid, then we can choose $\sigma_0 \to 0$, so that the error bound can approach 0 as well.

\begin{remark}
    The algorithm in \cite{Judi13} solves the compressed sensing subroutine within 1-norm error $\mathcal{O}(1)$ with runtime $\caO(N \poly\log N)$. In our numerical tests, we simply use the `Minimize' subroutine in CVXPY.
\end{remark}

\section{Numerical results}
\label{sec:comparison}

\subsection{Previous algorithms}

We begin with a brief review of three different early fault-tolerant QPE algorithms: ML-QCELS \cite{PRXQuantum.4.020331}, MM-QCELS \cite{ding2023simultaneous}, and QMEGS \cite{ding2024quantum}, with the latter two algorithms applicable to QEEP.

The outlines of the first two algorithms are as follows. The ML-QCELS algorithm features a hierarchical structure. In each hierarchy, Hadamard tests are used to estimate the signal $y(t)$ at $N_0$ different time points. The algorithm then outputs an estimate for the dominant frequency by minimizing the cost function:
\begin{equation}
L(r,E) = \frac{1}{N_0}\sum_{n=1}^{N_0}\left| re^{-\rmi En\tau} - y(n\tau)\right|^2.
\end{equation}
In the subsequent hierarchy, the search is refined to a narrower region, yielding a new estimate. This iterative process ultimately produces an accurate estimate of the dominant frequency. While ML-QCELS has proven effective for single-phase estimation, it is not suitable for multiple-phase estimation.

To address this limitation, the authors later introduced a multiple-phase version called MM-QCELS \cite{ding2023simultaneous}. This algorithm modifies the original ML-QCELS in two significant ways: the time samples $\mathcal{T}$ are drawn from a probability distribution, and the cost function is redefined as
\begin{equation}
    L\left(\{r_k, E_k\}_{k=1}^K\right) = \frac{1}{|\mathcal{T}|}\sum_{t\in\mathcal{T}}\left| \sum_{k=1}^K r_k e^{-\rmi E_k t} - y(t)\right|^2.
\end{equation}
Additionally, when applying MM-QCELS to single-phase estimation, the hierarchical structure can be eliminated, rendering the algorithm non-adaptive.

A more recent algorithm, QMEGS \cite{ding2024quantum}, also samples time from a continuous probability distribution. Instead of estimating the frequencies by minimizing $L\left(\{r_k, E_k\}_{k=1}^K\right)$, the algorithm identifies the dominant frequency of the target signal by solving
\begin{equation}
    \min_{f\in[0,1)} \sum_{t\in\mathcal{T}}\left|y(t)e^{\rmi 2\pi ft}\right|^2.
\end{equation}
Let the solution of this optimization be $f^\ast$. In the next step, the algorithm narrows its search to the region $[0,1]/(f^\ast - \delta f, f^\ast+\delta f)$ to estimate the sub-dominant frequency. By iterating this procedure, all dominant frequencies can eventually be determined.

Compressed sensing, in essence, is similar to non-adaptive MM-QCELS. Both methods aim to fit the sampled data using a signal ansatz. However, they differ in sampling strategies and optimization objectives. In MM-QCELS, times are sampled from a continuous probability distribution, and the cost function minimizes the total empirical error in the time domain. In contrast, compressed sensing samples times uniformly at random from a discrete set, and the cost function minimizes the 1-norm of the signal in the frequency domain.

\subsection{Models and results}

\begin{figure}
    \centering
    \includegraphics[width=0.9\linewidth]{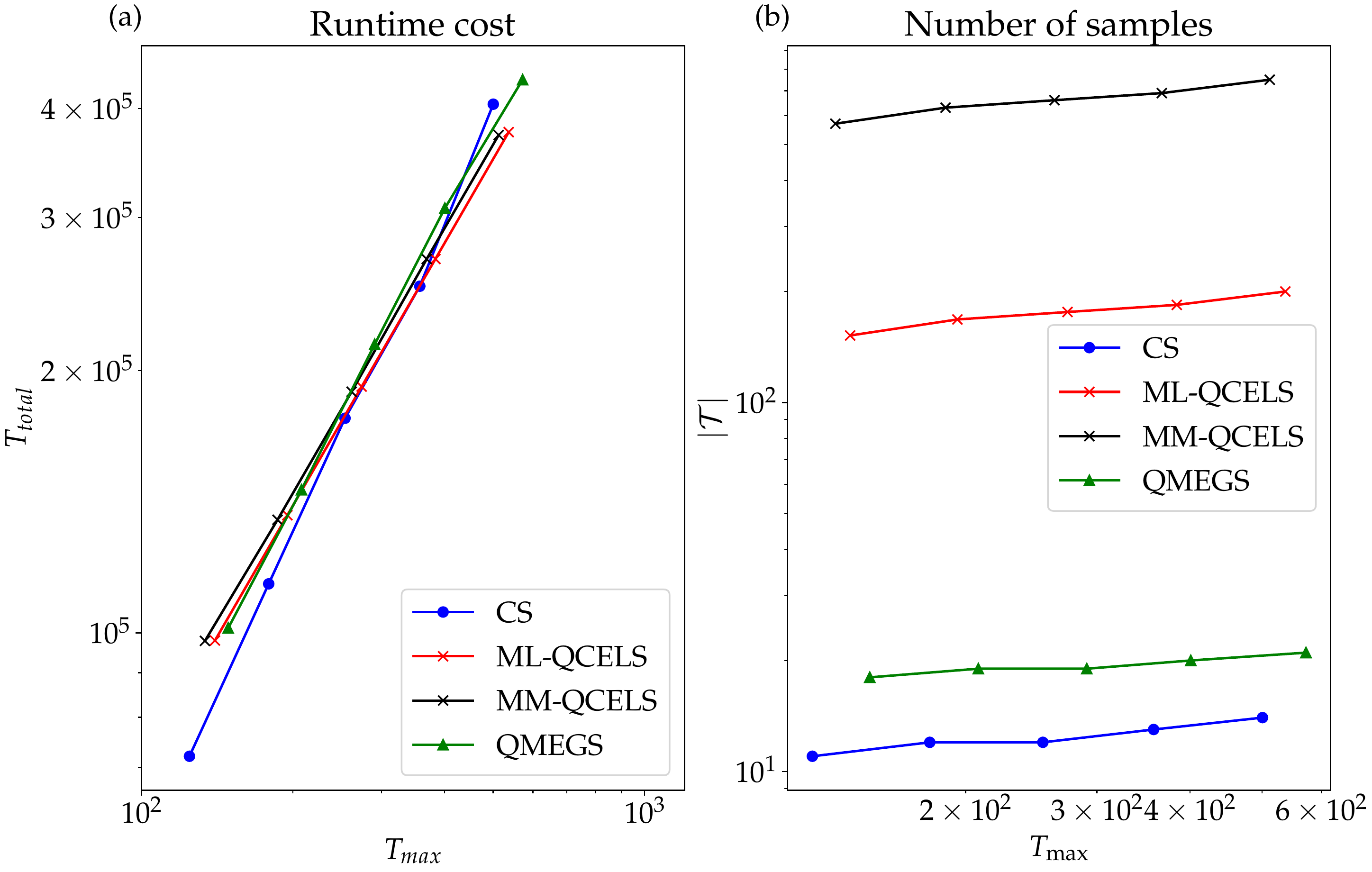}
    \caption{Comparison between the setups of the four algorithms. Here we set $T_n = \lfloor 100 \times (1.4)^{n}\rfloor, n = 1,2,3,4,5$. In Algorithm~\ref{algorithm:single_eigenvalue}, we set $S = 1, r = 2.3\ln T_n/T_n,\tau = 1, \sigma = 0.2\sqrt{2.3\ln T_n}, J = 100, M_{\mathrm{H}} = 100$ and let $\caK = \{\arg\max s_{\nu_\ast}\}$. In ML-QCELS, we set $N = 8, N_s = 50, J = \lfloor 4\ln T_n\rfloor$ (the meanings of the parameters can be found in \cite{PRXQuantum.4.020331}). In MM-QCELS, we set $K = 2, N_T = 30, \gamma = 1, N_s = 100, J = \lfloor 4\ln T_n\rfloor$ (the meanings of the parameters can be found in \cite{ding2023simultaneous}). In QMEGS, we set $K = 10, \mathrm{dx} = 10^{-4}, \alpha = 5, N = 10 + 2\lfloor\ln (2T_n)\rfloor$ (the meanings of the parameters can be found in \cite{ding2024quantum}). The codes are available online \cite{github}.}
    \label{fig:sample_number}
\end{figure}

\begin{figure}
    \centering
    \includegraphics[width = \textwidth]{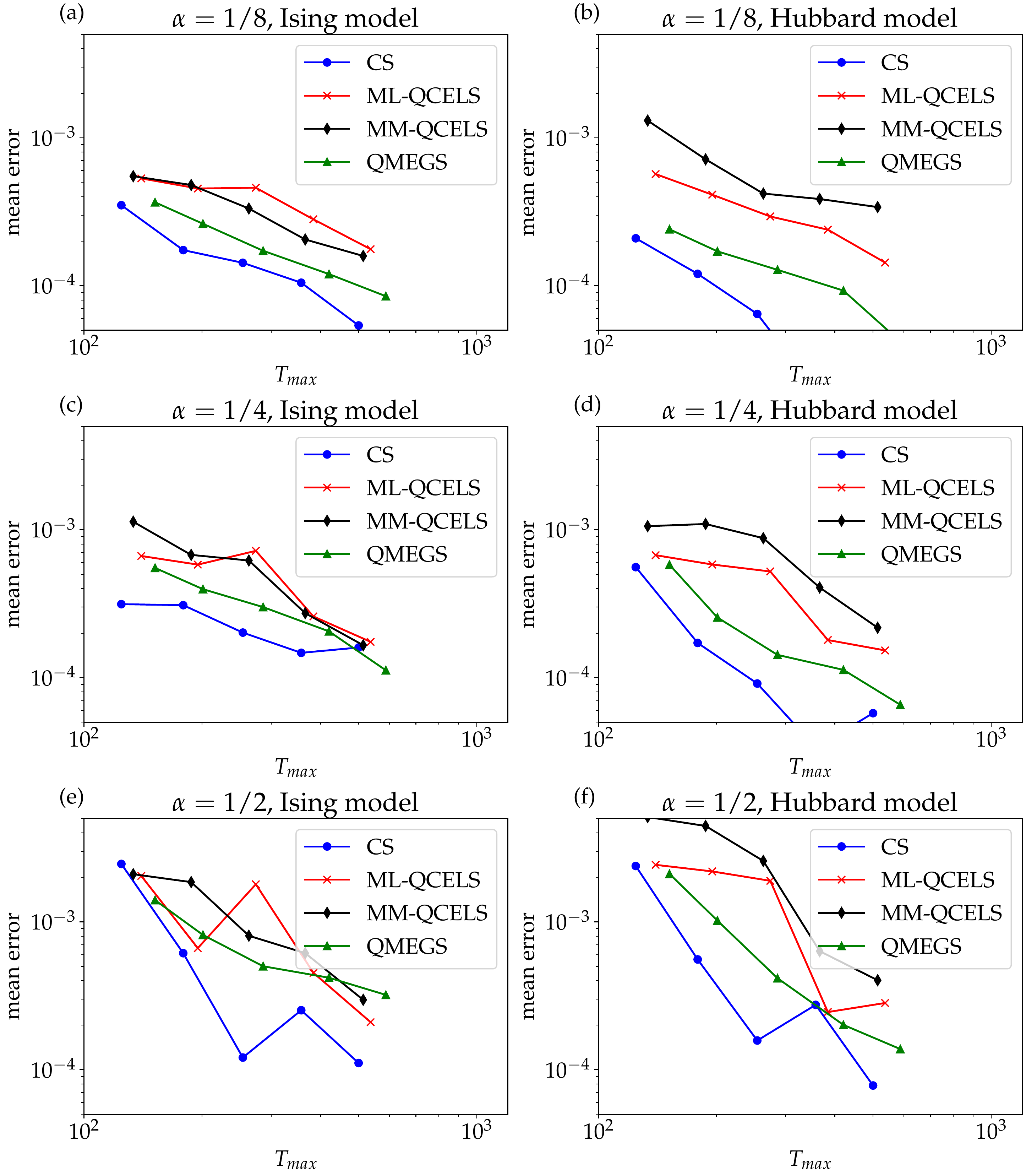}
    \caption{\label{fig:mean_error} Comparison between Algorithm~\ref{algorithm:single_eigenvalue},  ML-QCELS, MM-QCELS and QMEGS with respect to signals generated by Ising model (Eq.~(\ref{equ:ising}), left column) or Hubbard model (Eq.~(\ref{equ:hubbard}), right column). Subfigures (a),(b) represent $\alpha = 1/8$ ($p_0 \approx 0.88$); (c),(d) represent $\alpha = 1/4$ ($p_0 \approx 0.75$); (e),(f) represent $\alpha = 1/2$ ($p_0 \approx 0.5$).}
    \label{fig:qpe_ising}
\end{figure}

In this subsection, we present several numerical tests and compare with previous works. 

Given a Hamiltonian $H$, we first renormalize it to
\begin{equation}
    \overline{H} = \frac{\pi H}{4\|H\|_2}
\end{equation}
so that the spectra of $\overline{H}$ belong to $[-\pi/4, \pi/4]$. We further shift the Hamiltonian to $\overline{H}' = \overline{H} + \pi/2$ so that the spectrum of $\overline{H}'$ belong to $[\pi/4,3\pi/4] \subset [0,2\pi]$, as assumed in the initial setup of our algorithm. To have a better control of the parameters in the signal, we design the following family of initial states: choose a parameter $\alpha \in (0,1)$, and set the initial state as
\begin{equation}
    |\Psi_\alpha\rangle \propto \sum_{\ell=0}^{9} \sqrt{\alpha^{\ell}}|\phi_\ell\rangle,
\end{equation}
where $|\phi_\ell\rangle$ is the $\ell$-th eigenstate of $\overline{H}$ with eigenvalue $E_\ell$. Then, the target signal is as follows:
\begin{equation}
    y_n = \sum_{\ell=0}^{9} p_{\ell} e^{-\rmi  E_\ell n \tau},\quad p_\ell = \frac{(1-\alpha)\alpha^\ell}{1 - \alpha^{10}}.
\end{equation}
For example, when $\alpha = 1/2$, we have $p_0 \approx 1/2$. Clearly, $p_0$ decreases with $\alpha$.

In the first set of numerical tests, $H$ is the normalized Transverse Field Ising (TFI) model on 8 sites:
\begin{equation}
    \quad H_{\mathrm{TFI}} = -\sum_{j=1}^7 Z_j Z_{j+1} - Z_8 Z_1 - 4\sum_{j=1}^8 X_j.
\label{equ:ising}
\end{equation}
In the second set of numerical tests, $H$ is the Fermi-Hubbard (FH) model on 4 sites:
\begin{equation}  \label{equ:hubbard}
    H_{\mathrm{FH}} = -\sum_{j=1}^{3}\sum_{\sigma = \uparrow,\downarrow}c^\dagger_{j,\sigma}c_{j+1,\sigma} + 10\sum_{j=1}^4 \left(n_{j,\uparrow} - \frac{1}{2}\right)\left(n_{j,\downarrow} - \frac{1}{2}\right).
\end{equation}
For each Hamilonian, we set $\alpha = 1/8, 1/4, 1/2$ separately. 

To have a rather fair comparison between the algorithms, we deliberately choose the parameters to ensure that the runtimes $(T_{\max}, T_{\text{total}})$ of the algorithms are approximately on the same line, which is demonstrated in Fig.~\ref{fig:sample_number}. In the same figure, we also plot the number of distinct $t$ used in these algorithms. As shown in Fig.~\ref{fig:sample_number}, our algorithm requires a much smaller size of time samples $|\mathcal{T}|$, in contrast with that of ML-QCELS and MM-QCELS. The mean errors of the outputs are shown in Fig.~\ref{fig:mean_error}. When the initial overlap is comparably large $(\alpha = 1/8)$, the performance of compressed sensing based algorithm is better than the other methods. This observation supports the claim that Algorithm~\ref{algorithm:single_eigenvalue} is prominent in its sparse sampling of time and high level of accuracy.

\section{Discussions}
\label{sec:discussion}

In this paper, we presented a simple and robust algorithm for QPE using compressed sensing. For the single eigenvalue estimation (i.e., QPE), we rigorously established its Heisenberg-limit scaling in Theorem~\ref{theorem:main} and numerically demonstrated its performance compared to other state-of-the-art QPE algorithms in Sec.~\ref{sec:comparison}. Our algorithm has a smaller average error when the initial overlap is large, provided that the runtime costs $(T_{\mathrm{total}},T_{\mathrm{max}})$ are approximately the same. Similarly to QMEGS, our algorithm is non-adaptive, which means we can perform all the measurements first, and then focus on the classical post-processing part. As a comparison, RPE-inspired algorithms are usually adaptive \cite{PRXQuantum.4.020331,ni2023low}. Our algorithm requires a rather small size of $\mathcal{T}$ on a discrete set. In Fig.~\ref{fig:qpe_ising}, for a signal of length $N \in (100, 600)$, our algorithm only requires $\approx 10^1$ different time samples while MM-QCELS requires $\approx 10^2$ number of time samples from a continuous region. Our numerical tests also show that QMEGS only requires $\approx 10^1$ different samples to obtain an accurate estimation. As mentioned in the end of Sec.~\ref{sec:mainidea}, there is a close connection between QMEGS and algorithm OMP, and we conjecture that it is possible to analyze QMEGS using the framework of compressed sensing. 

Lastly, we list a few open questions:
\begin{itemize}
    \item In discrete sampling protocols, would it be possible to shorten the maximal runtime by biased sampling of times? What is the limitation in the discrete scenario? Can we achieve a similar improvement to the Gaussian filter method in \cite{ding2024quantum}? 
    \item In our numerical experiments, the test of another sampling (Algorithm~\ref{algorithm:test}) is actually unnecessary. Is it possible to show this analytically as well?
    \item It is possible to find the optimal grid shift parameter by optimization instead of doing brute-force search in a trial set, which should significantly improve our current result.
\end{itemize}

\section{Acknowledgement}

We thank Tianyu Wang for the helpful discussions. C.Y. acknowledges support from the National Natural
Science Foundation of China (Grant No.~92165109),
National Key Research and Development Program
of China (Grant No.~2022YFA1404204), and Shanghai
Municipal Science and Technology Major Project
(Grant No.~2019SHZDZX01). C.Z. and J.T. acknowledge support from the U.S. National Science Foundation under Grant No. 2116246, the U.S. Department of Energy, Office of Science, National Quantum Information Science Research Centers, and Quantum Systems Accelerator.

\bibliographystyle{quantum}
\bibliography{main}

\begin{thebibliography}{10}

\bibitem{kitaev1995quantum}
Alexei Y~Kitaev.
\newblock ``Quantum measurements and the {Abelian} stabilizer problem''.
\newblock \href{https://dx.doi.org/10.48550/arXiv.quant-ph/9511026}{quant-ph/9511026}~(1995).

\bibitem{shor1999polynomial}
Peter~W Shor.
\newblock ``Polynomial-time algorithms for prime factorization and discrete logarithms on a quantum computer''.
\newblock \href{https://dx.doi.org/10.1137/S0097539795293172}{SIAM review {\bf 41}, 303--332}~(1999).

\bibitem{abrams1999quantum}
Daniel~S Abrams and Seth Lloyd.
\newblock ``Quantum algorithm providing exponential speed increase for finding eigenvalues and eigenvectors''.
\newblock \href{https://dx.doi.org/10.1103/PhysRevLett.83.5162}{Phys. Rev. Lett. {\bf 83}, 5162}~(1999).

\bibitem{mcardle2020quantum}
Sam McArdle, Suguru Endo, Al{\'a}n Aspuru-Guzik, Simon~C Benjamin, and Xiao Yuan.
\newblock ``Quantum computational chemistry''.
\newblock \href{https://dx.doi.org/10.1103/RevModPhys.92.015003}{Rev. Mod. Phys. {\bf 92}, 015003}~(2020).

\bibitem{lin2022heisenberg}
Lin Lin and Yu~Tong.
\newblock ``Heisenberg-limited ground-state energy estimation for early fault-tolerant quantum computers''.
\newblock \href{https://dx.doi.org/10.1103/PRXQuantum.3.010318}{PRX Quantum {\bf 3}, 010318}~(2022).

\bibitem{wang2022quantum}
Guoming Wang, Daniel Stilck-Fran{\c{c}}a, Ruizhe Zhang, Shuchen Zhu, and Peter~D Johnson.
\newblock ``Quantum algorithm for ground state energy estimation using circuit depth with exponentially improved dependence on precision''.
\newblock \href{https://dx.doi.org/10.22331/q-2023-11-06-1167}{Quantum {\bf 7}, 1167}~(2023).

\bibitem{somma2019quantum}
Rolando~D Somma.
\newblock ``Quantum eigenvalue estimation via time series analysis''.
\newblock \href{https://dx.doi.org/10.1088/1367-2630/ab5c60}{New J. Phys. {\bf 21}, 123025}~(2019).

\bibitem{o2019quantum}
Thomas~E O’Brien, Brian Tarasinski, and Barbara~M Terhal.
\newblock ``Quantum phase estimation of multiple eigenvalues for small-scale (noisy) experiments''.
\newblock \href{https://dx.doi.org/10.1088/1367-2630/aafb8e}{New J. Phys. {\bf 21}, 023022}~(2019).

\bibitem{zhang2022computing}
Ruizhe Zhang, Guoming Wang, and Peter Johnson.
\newblock ``Computing ground state properties with early fault-tolerant quantum computers''.
\newblock \href{https://dx.doi.org/10.22331/qv-2022-07-22-65}{Quantum {\bf 6}, 761}~(2022).

\bibitem{dut22hei}
Alicja Dutkiewicz, Barbara~M Terhal, and Thomas~E O’Brien.
\newblock ``Heisenberg-limited quantum phase estimation of multiple eigenvalues with few control qubits''.
\newblock \href{https://dx.doi.org/10.22331/q-2022-10-06-830}{Quantum {\bf 6}, 830}~(2022).

\bibitem{nielsen2010quantum}
Michael~A Nielsen and Isaac~L Chuang.
\newblock ``Quantum computation and quantum information''.
\newblock \href{https://dx.doi.org/10.1017/CBO9780511976667}{Cambridge university press}. ~(2010).

\bibitem{PRXQuantum.4.020331}
Zhiyan Ding and Lin Lin.
\newblock ``Even shorter quantum circuit for phase estimation on early fault-tolerant quantum computers with applications to ground-state energy estimation''.
\newblock \href{https://dx.doi.org/10.1103/PRXQuantum.4.020331}{PRX Quantum {\bf 4}, 020331}~(2023).

\bibitem{ni2023low}
Hongkang Ni, Haoya Li, and Lexing Ying.
\newblock ``On low-depth algorithms for quantum phase estimation''.
\newblock \href{https://dx.doi.org/10.22331/q-2023-11-06-1165}{Quantum {\bf 7}, 1165}~(2023).

\bibitem{georgescu2014quantum}
Iulia~M Georgescu, Sahel Ashhab, and Franco Nori.
\newblock ``Quantum simulation''.
\newblock \href{https://dx.doi.org/10.1103/RevModPhys.86.153}{Rev. Mod. Phys. {\bf 86}, 153}~(2014).

\bibitem{childs2021theory}
Andrew~M Childs, Yuan Su, Minh~C Tran, Nathan Wiebe, and Shuchen Zhu.
\newblock ``Theory of {T}rotter error with commutator scaling''.
\newblock \href{https://dx.doi.org/10.1103/PhysRevX.11.011020}{Phys. Rev. X {\bf 11}, 011020}~(2021).

\bibitem{ding2023simultaneous}
Zhiyan Ding and Lin Lin.
\newblock ``Simultaneous estimation of multiple eigenvalues with short-depth quantum circuit on early fault-tolerant quantum computers''.
\newblock \href{https://dx.doi.org/10.22331/q-2023-10-11-1136}{Quantum {\bf 7}, 1136}~(2023).

\bibitem{PhysRevA.108.062408}
Haoya Li, Hongkang Ni, and Lexing Ying.
\newblock ``Adaptive low-depth quantum algorithms for robust multiple-phase estimation''.
\newblock \href{https://dx.doi.org/10.1103/PhysRevA.108.062408}{Phys. Rev. A {\bf 108}, 062408}~(2023).

\bibitem{ding2024quantum}
Zhiyan Ding, Haoya Li, Lin Lin, HongKang Ni, Lexing Ying, and Ruizhe Zhang.
\newblock ``Quantum {M}ultiple {E}igenvalue {G}aussian filtered {S}earch: an efficient and versatile quantum phase estimation method''.
\newblock \href{https://dx.doi.org/10.22331/q-2024-10-02-1487}{Quantum {\bf 8}, 1487}~(2024).

\bibitem{arad2016connecting}
Itai Arad, Tomotaka Kuwahara, and Zeph Landau.
\newblock ``Connecting global and local energy distributions in quantum spin models on a lattice''.
\newblock \href{https://dx.doi.org/10.1088/1742-5468/2016/03/033301}{J. Stat. Mech. Theor. Exp. {\bf 2016}, 033301}~(2016).

\bibitem{childs2019nearly}
Andrew~M Childs and Yuan Su.
\newblock ``Nearly optimal lattice simulation by product formulas''.
\newblock \href{https://dx.doi.org/10.1103/PhysRevLett.123.050503}{Phys. Rev. Lett. {\bf 123}, 050503}~(2019).

\bibitem{hassanieh2012nearly}
Haitham Hassanieh, Piotr Indyk, Dina Katabi, and Eric Price.
\newblock ``Nearly optimal sparse {F}ourier transform''.
\newblock In Proceedings of the forty-fourth annual ACM symposium on Theory of computing.
\newblock \href{https://dx.doi.org/10.1145/2213977.2214029}{Pages 563--578}.
\newblock ~(2012).

\bibitem{liao2016music}
Wenjing Liao and Albert Fannjiang.
\newblock ``{MUSIC} for single-snapshot spectral estimation: {S}tability and super-resolution''.
\newblock \href{https://dx.doi.org/10.1016/j.acha.2014.12.003}{Appl. Comput. Harmon. Anal. {\bf 40}, 33--67}~(2016).

\bibitem{song2023quartic}
Zhao Song, Baocheng Sun, Omri Weinstein, and Ruizhe Zhang.
\newblock ``Quartic samples suffice for {F}ourier interpolation''.
\newblock In 2023 IEEE 64th Annual Symposium on Foundations of Computer Science (FOCS).
\newblock \href{https://dx.doi.org/FOCS57990.2023.00087}{Pages 1414--1425}.
\newblock IEEE~(2023).

\bibitem{ding2024esprit}
Zhiyan Ding, Ethan~N Epperly, Lin Lin, and Ruizhe Zhang.
\newblock ``The {ESPRIT} algorithm under high noise: {Optimal} error scaling and noisy super-resolution''.
\newblock In 2024 IEEE 65th Annual Symposium on Foundations of Computer Science (FOCS).
\newblock \href{https://dx.doi.org/10.1109/FOCS61266.2024.00137}{Pages 2344--2366}.
\newblock IEEE~(2024).

\bibitem{cochran1967fast}
William~T Cochran, James~W Cooley, David~L Favin, Howard~D Helms, Reginald~A Kaenel, William~W Lang, George~C Maling, David~E Nelson, Charles~M Rader, and Peter~D Welch.
\newblock ``What is the fast {F}ourier transform?''.
\newblock \href{https://dx.doi.org/10.1109/PROC.1967.5957}{Proceedings of the IEEE {\bf 55}, 1664--1674}~(1967).

\bibitem{gilbert2005improved}
Anna~C Gilbert, Shan Muthukrishnan, and Martin Strauss.
\newblock ``{Improved time bounds for near-optimal sparse Fourier representations}''.
\newblock In Manos Papadakis, Andrew~F. Laine, and Michael~A. Unser, editors, Wavelets XI.
\newblock \href{https://dx.doi.org/10.1117/12.615931}{Volume 5914, page 59141A}.
\newblock International Society for Optics and PhotonicsSPIE~(2005).

\bibitem{indyk2014nearly}
Piotr Indyk, Michael Kapralov, and Eric Price.
\newblock ``{(Nearly)} sample-optimal sparse {F}ourier transform''.
\newblock In Proceedings of the twenty-fifth annual ACM-SIAM symposium on Discrete algorithms.
\newblock \href{https://dx.doi.org/10.1109/FOCS.2019.00092}{Pages 480--499}.
\newblock SIAM~(2014).

\bibitem{higgins2009demonstrating}
Brendon~L Higgins, Dominic~W Berry, Stephen~D Bartlett, Morgan~W Mitchell, Howard~M Wiseman, and Geoff~J Pryde.
\newblock ``Demonstrating {H}eisenberg-limited unambiguous phase estimation without adaptive measurements''.
\newblock \href{https://dx.doi.org/10.1088/1367-2630/11/7/073023}{New J. Phys. {\bf 11}, 073023}~(2009).

\bibitem{kimmel2015robust}
Shelby Kimmel, Guang~Hao Low, and Theodore~J Yoder.
\newblock ``Robust calibration of a universal single-qubit gate set via robust phase estimation''.
\newblock \href{https://dx.doi.org/10.1103/PhysRevA.92.062315}{Phys. Rev. A {\bf 92}, 062315}~(2015).

\bibitem{Bel20ach}
Federico Belliardo and Vittorio Giovannetti.
\newblock ``Achieving {H}eisenberg scaling with maximally entangled states: {A}n analytic upper bound for the attainable root-mean-square error''.
\newblock \href{https://dx.doi.org/10.1103/physreva.102.042613}{Phys. Rev. A{\bf 102}}~(2020).

\bibitem{cai2011orthogonal}
T~Tony Cai and Lie Wang.
\newblock ``Orthogonal matching pursuit for sparse signal recovery with noise''.
\newblock \href{https://dx.doi.org/10.1109/TIT.2011.2146090}{IEEE Transactions on Information theory {\bf 57}, 4680--4688}~(2011).

\bibitem{candes2006near}
Emmanuel~J Cand{\`e}s and Terence Tao.
\newblock ``Near-optimal signal recovery from random projections: {U}niversal encoding strategies?''.
\newblock \href{https://dx.doi.org/10.1109/TIT.2006.885507}{IEEE transactions on information theory {\bf 52}, 5406--5425}~(2006).

\bibitem{candes2008restricted}
Emmanuel~J Cand{\`e}s.
\newblock ``The restricted isometry property and its implications for compressed sensing''.
\newblock \href{https://dx.doi.org/10.1016/j.crma.2008.03.014}{Comptes rendus. Mathematique {\bf 346}, 589--592}~(2008).

\bibitem{candes2006robust}
Emmanuel~J Cand{\`e}s, Justin Romberg, and Terence Tao.
\newblock ``Robust uncertainty principles: {E}xact signal reconstruction from highly incomplete frequency information''.
\newblock \href{https://dx.doi.org/10.1109/TIT.2005.862083}{IEEE Transactions on information theory {\bf 52}, 489--509}~(2006).

\bibitem{gross2010quantum}
David Gross, Yi-Kai Liu, Steven~T Flammia, Stephen Becker, and Jens Eisert.
\newblock ``Quantum state tomography via compressed sensing''.
\newblock \href{https://dx.doi.org/10.1103/PhysRevLett.105.150401}{Phys. Rev. Lett. {\bf 105}, 150401}~(2010).

\bibitem{magesan2013compressing}
Easwar Magesan, Alexandre Cooper, and Paola Cappellaro.
\newblock ``Compressing measurements in quantum dynamic parameter estimation''.
\newblock \href{https://dx.doi.org/10.1103/PhysRevA.88.062109}{Phys. Rev. A {\bf 88}, 062109}~(2013).

\bibitem{Smith2013}
Aaron Smith, Riofr{\'\i}o Carlos, Brielle~Evelyn Anderson, Hector~Sosa Martinez, Ivan~H Deutsch, and Poul Jessen.
\newblock ``Quantum state tomography by continuous measurement and compressed sensing''.
\newblock \href{https://dx.doi.org/10.1103/PhysRevA.87.030102}{Phys. Rev. A {\bf 87}, 030102}~(2013).

\bibitem{Kalev2015}
Amir Kalev, Robert~L Kosut, and Ivan~H Deutsch.
\newblock ``Quantum tomography protocols with positivity are compressed sensing protocols''.
\newblock \href{https://dx.doi.org/10.1038/npjqi.2015.18}{Npj Quantum Inf. {\bf 1}, 15018}~(2015).

\bibitem{tang2013compressed}
Gongguo Tang, Badri~Narayan Bhaskar, Parikshit Shah, and Benjamin Recht.
\newblock ``Compressed sensing off the grid''.
\newblock \href{https://dx.doi.org/10.1109/TIT.2013.2277451}{IEEE transactions on information theory {\bf 59}, 7465--7490}~(2013).

\bibitem{Judi13}
Juditsky Anatoli, Kilinc~Karzan Fatma, and Nermirovski Arkadi.
\newblock ``Randomized first order algorithms with applications to $\ell_1$-minimization''.
\newblock \href{https://dx.doi.org/10.1007/s10107-012-0575-2}{Math. Program. {\bf 142}, 269--310}~(2013).

\bibitem{github}
\url{https://github.com/CYI1995/QEEP/tree/main/Paper_QPE}.

\bibitem{rudelson2008sparse}
Mark Rudelson and Roman Vershynin.
\newblock ``On sparse reconstruction from {F}ourier and {G}aussian measurements''.
\newblock \href{https://dx.doi.org/10.1002/cpa.20227}{Comm. Pure Appl. Math. {\bf 61}, 1025--1045}~(2008).

\end{thebibliography}

\newpage

\appendix

\section{Standard results in compressed sensing} \label{app:compressed_sensing}

Given a vector $v = [v_1, v_2, \cdots, v_N]^{\top}$, its $p$-norm is defined as
\begin{equation}
    \|v\|_p \equiv \left(\sum_{n=1}^N |v_n|^p\right)^{1/p}.
\end{equation}
Note that $\|v\|_{\infty} = \max_n|v_n|$ and $\|v\|_0$ is the sparsity of $v$. In the following paragraph, we use $k$ to label the indices of entries in the frequency domain, and use $n$ to label the indices of entries in the time domain. Denote the set of integers from 1 to $N$ as $[N]$. In regular compressed sensing, we deal with a time domain discrete signal $y^0$ in the form of
\begin{equation}
    y^0_n = \sum_{f\in\mathcal{F}}p_f e^{-\rmi 2\pi fn}\quad \forall n\in [N],
\label{equ:time_domain_signal2}
\end{equation}
where $p_f > 0, \sum_f p_f = 1, f\in [0,1)$, and $\mathcal{F}$ is the set of frequencies. In the context of compressed sensing, sparsity means $|\mathcal{F}| = \mathcal{O}(\log N)$ \cite{candes2006near}. The time domain signal can thus be written as an $N$-dimensional vector
\begin{equation}
    y^0 = [y^0_1, y^0_2, \cdots, y^0_N]^{\top}.
\end{equation}
Define the Fourier matrix by $F_{kn} \equiv e^{-\rmi 2\pi kn/N}$ with $k,n\in[N]$. If all $f\in\mathcal{F}$ are on-grid, the frequency domain signal $x$ can be written in the form of a real vector:
\begin{equation}
    x = \frac{1}{N}F^{\dagger}y^0 = \sum_{f\in\mathcal{F}}p_f \delta_{Nf}.
\end{equation}
The purpose of compressed sensing is to recover $x$ from noisy signal samples. The algorithm is accomplished in the following sequence. Choose a sampling ratio $r\in (0,1)$, and assign each integer $n$ in $[N]$ a random variable $1_n$ such that
\begin{equation}
    \mathrm{Pr}\{1_n = 1\} = r,\quad \mathrm{Pr}\{1_n = 0\} = 1-r.
\label{equ:indicator}
\end{equation}
Draw one sample from each $1_n$, and denote the set of integers with sampled value $1$ as the sample set $\mathcal{T}$. Given $\mathcal{T}$, we define the projection operator $\mathcal{P}_\mathcal{T}$ by
\begin{equation}
    (\mathcal{P}_\mathcal{T})_{t_1,t_2} = 1_{t_1\in\mathcal{T}}\cdot \delta_{t_1,t_2}\quad t_1,t_2\in[N],
\end{equation}
and $F_{\mathcal{T}} \equiv \mathcal{P}_\mathcal{T} F, y^0_{\mathcal{T}} \equiv \mathcal{P}_\mathcal{T} y^0$.

With these notations, the compressed sensing subroutine is to solve the following optimization problem
\begin{equation}
    \min_{s\in\mathbb{R}^N} \|s\|_1,\quad \mathrm{s.t.}\quad F_{\mathcal{T}}s = y^0_{\mathcal{T},}
\label{equ:cs1}
\end{equation}
which can be rewritten as a linear programming problem. When $|\mathcal{T}| \approx Nr = \mathcal{O}(\poly\log N)$ and the frequency support $\mathcal{F}$ is sparse in the sense that $|\mathcal{F}| = \mathcal{O}(\poly\log N)$, the optimal solution $s^\#$ is equal to the frequency domain signal $x$ with high probability. Rigorous statements can be found in \cite{candes2006near}. 

Provided that the signal has extra noise $y_n = y^0_n + z_n$, then the signal can be approximately recovered by the convex relaxation algorithm \cite{candes2008restricted}:
\begin{equation}
    \min_{s\in\mathbb{R}^N} \|s\|_1,\quad \mathrm{s.t.}\quad \|F_{\mathcal{T}}s - y_{\mathcal{T}}\|_2 \le \sqrt{|\mathcal{T}|}\sigma,
\label{equ:convex_relaxation}
\end{equation}
where $\sigma$ is the expected mean-square-root noise. The differences between the solution to Eq.~(\ref{equ:convex_relaxation}) and $x$ depend on $\sigma$. The subroutine itself is a convex quadratic programming problem.

The robustness of compressed sensing solution can be analyzed through the restricted isometry property (RIP) of random Fourier matrices. If a matrix $M$ satisfies $(1-\eta)\|x\|_2^2 \le \|Mx\|_2^2 \le (1+\eta)\|x\|_2^2$ for all $x\in\mathcal{X}$, then we say that the matrix $M$ satisfies $\eta$-RIP over set $\mathcal{X}$. We present a few standard results of compressed sensing in the following. 

\begin{theorem}\cite{rudelson2008sparse} \label{theorem:rip}
    Suppose $\mathcal{T}$ is generated by sampling from $[N]$ with ratio 
    \begin{equation}
        r = \mathcal{O}\left(S \log^2 (S) \eta^{-2} \mathrm{poly}\log(N) N^{-1}\right).
    \end{equation}
    Then $F_\mathcal{T}/\sqrt{|\mathcal{T}|}$ satisfies $\eta$-RIP over the set
       $\{x : x\in \mathbb{C}^N, \|x\|_0 \le S\}$
    with probability
    \begin{equation}
        1 - \mathcal{O}\left(\frac{\mathrm{poly}(Se^{\eta^{-1}})}{\mathrm{poly}(N)}\right).
    \end{equation}
\end{theorem}
\begin{proof}
    Follow Theorem 3.11 of \cite{rudelson2008sparse}. The theorem can be proved by setting $s = \mathcal{O}(\log N)$ and $\varepsilon = \mathcal{O}(\eta/\log N)$. The meaning of the notations can be found in the same paper.
\end{proof}

Given a vector $v\in \mathbb{R}^N$, we rearrange the entries of the vector by their absolute values in the non-increasing order to obtain $(v^{\downarrow}_1, v^{\downarrow}_2, \cdots, v^{\downarrow}_N)$. Define $v_{\mathrm{res}}$ as the vector generated from $v$ by removing $v^{\downarrow}_1, v^{\downarrow}_2, \ldots, v^{\downarrow}_S$. The accuracy of Eq.~(\ref{equ:convex_relaxation}) can be estimated by the following theorem.

\begin{theorem}\cite{candes2008restricted} \label{theorem:compressed_sensing}
    Suppose $M\in \mathbb{C}^{N'\times N}$ satisfies $\eta$-RIP over the set $\{x : x\in \mathbb{R}^N, \|x\|_0 \le 2S\}$, $0 < \eta < \sqrt{2}-1$, and $x_1,x_2\in \mathbb{R}^N$. If
    \begin{equation}
        \|x_1\|_1 \le \|x_2\|_1,\quad \|M(x_1-x_2)\|_2 \le \sigma,
    \end{equation}
    then
    \begin{gather}
        \|x_1 - x_2\|_2 \le C_1 \sigma + C_2\frac{\|x_{2,\mathrm{res}}\|_1}{\sqrt{S}},
    \end{gather}
    where \begin{equation} \label{equ:C1C2}
        C_1 \equiv \frac{2 + 2(\sqrt{2}-1)\eta}{1 - (\sqrt{2}+1)\eta},\quad C_2 \equiv \frac{4\sqrt{1 + \eta}}{1 - (\sqrt{2}+1)\eta}.
    \end{equation}
\end{theorem}

\section{Proof of Theorem~\ref{theorem:main}}
\label{app:proof_of_main_theorem}

The input signal is
\begin{equation} 
    y = y^0_{\on} + y^0_{\off} + z,
\end{equation}
where $z$ is the uncertainty from the Hadamard tests that satisfies $\|z\|_{\infty} \le \sH$. Recall that $s_\nu$ is the solution of
\begin{equation}\label{equ:compressed_sensing}
    \min_{s\in\mathbb{R}^N}\|s\|_1\quad \mathrm{s.t.}\quad \|F_{\nu,\mathcal{T}}s - y_{\mathcal{T}}\|_2 \le \sqrt{|\mathcal{T}|}\sigma.
\end{equation}
Let $\nu_1$ be the parameter in $\mathcal{V}$ that is closest to $\mathfrak{u}$, and $\nu_\ast$ be the one with the smallest $\|s_{\nu_\ast}\|_1$. 
Therefore,
\begin{align}
     \|F_{{\nu_1},\mathcal{T}}(s_{\nu_\ast} - s_{\nu_1})\|_2 &\le \|F_{{\nu_1},\mathcal{T}} s_{\nu_\ast} - y_\mathcal{T}\|_2 + \|F_{{\nu_1},\mathcal{T}} s_{\nu_1} - y_\mathcal{T}\|_2\nonumber\\
     &\le \|F_{\nu_1,\mathcal{T}}s_{\nu_\ast} - y_\mathcal{T}\|_2 + \sqrt{|\mathcal{T}|}\sigma\nonumber\\
     &\le \|F_{\nu_1,\mathcal{T}}s_{\nu_\ast} - F_{{\nu_\ast},\mathcal{T}} s_{\nu_\ast}\|_2 + \|F_{\nu_\ast,\mathcal{T}}s_{\nu_\ast} - y_\mathcal{T}\|_2 + \sqrt{|\mathcal{T}|}\sigma\nonumber\\
     &\le \|F_{\nu_1,\mathcal{T}}s_{\nu_\ast} - F_{{\nu_\ast},\mathcal{T}} s_{\nu_\ast}\|_2 + 2\sqrt{|\mathcal{T}|}\sigma,\\
    \|F_{\nu_\ast,\mathcal{T}}s_{\nu_\ast} - F_{{\nu_1},\mathcal{T}} s_{\nu_\ast}\|_2^2 &= \sum_{t\in\mathcal{T}}\left|\sum_{k\in[N]}s_{\nu_\ast,k} \left(e^{-\rmi 2\pi (k + \nu_\ast)t/N} - e^{-\rmi 2\pi (k + \nu_1)t/N}\right)\right|^2\nonumber\\
    &= \sum_{t\in\mathcal{T}}\left|\sum_{k\in[N]} s_{\nu_\ast,k}e^{-\rmi 2\pi kt/N}\right|^2 4\sin^2\left[\frac{\pi(\nu_{\ast} - \nu_1)t}{N}\right]\nonumber\\
    &= \sum_{t\in\mathcal{T}}|(F s_{\nu_\ast})_t|^2 4\sin^2\left[\frac{\pi(\nu_{\ast} - \nu_1)t}{N}\right]\nonumber\\
    &\le 4\pi^2|\mathcal{T}| |\nu_{\ast} - \nu_1|^2 \quad \textrm{with high probability}.
\end{align}
The last inequality holds with high probability because $|\caT| = \mathcal{O}(\poly\log N)$. 

Combining Eqs.~(\ref{equ:lemma6nulowerbound}) and (\ref{equ:sigmatestlowerbound}) in Lemma~\ref{lemma:bad_nu}, if we set
\begin{equation}
    \sigma_\test = c \sqrt{\frac{3}{2}}\left[\left(C_3^2 + \frac{8\pi^2 S}{3\log^2 N}\right)^{1/2} + C_0\right]\sigma,
\end{equation}
where $c$ is a small constant larger than 1, then with high probability, we have
\begin{align}
    |\nu_\ast - \mathfrak{u}| \le C[x_{\on}]^{-1}(\sqrt{2}\sigma_\test + C_0\sigma) = \caO(\sigma).
\end{align}
The condition $\|x_{\on}\|_2 \gg \|x_{\on}\|_1/\sqrt{N}$ ensures that $C[x_{\on}] = \caO(1)$. Therefore,
\begin{equation}
    \begin{aligned}
        \|F_{{\nu_1},\mathcal{T}}(s_{\nu_\ast} - s_{\nu_1})\|_2 \le 2\sqrt{|\mathcal{T}|}\sigma + 2\pi|\nu_1 - \fku| + 2\pi|\nu_\ast - \fku|:= \sqrt{|\caT|}C_4\sigma_0.
    \end{aligned}
\end{equation}
The condition on $J$ in Eq.~(\ref{equ:sigmas}) gives us
\begin{equation} \label{equ:Jbound}
    |\nu_1 - \mathfrak{u}| \le J^{-1} \le \frac{\sqrt{S}\sigma_0}{\log N}.
\end{equation}
The conditions in Eq.~(\ref{equ:sigmas}) implies that $\sqrt{S}/\log N = \caO(1), C_0 = \caO(1), C_3 = \caO(1)$, which further implies $C_4 = \caO(1)$.

Consider the grid decomposition of $y^0_{\mathrm{on}}$ with respect to $\nu$:
\begin{equation} \label{equ:gird_decomposition}
    x^{\mathrm{R}}_\nu \equiv \mathrm{Re}(F^{-1}_\nu y^0_{\mathrm{on}}) ,\quad x^{\mathrm{I}}_\nu \equiv \mathrm{Im}(F^{-1}_\nu y^0_{\mathrm{on}}),
\end{equation}
Note that $x^{\mathrm{R}}_{\mathfrak{u}} = x_{\mathrm{on}}, x^{\mathrm{I}}_{\mathfrak{u}} = 0$. By virtue of Theorem~\ref{theorem:compressed_sensing}, we have
\begin{equation}
    \|s_{\nu_\ast} - s_{\nu_1}\|_2 \le C_1 C_4\sigma + C_2\frac{\|s_{\nu_1,\mathrm{res}}\|_1}{\sqrt{S}},
\end{equation}
where we have used the fact that $C_1,C_2,C_4 = \caO(1)$. The condition on $p_{\min}$ ensures that the dominant part of $s_{\nu_1}$ is the same as the dominant part of $x_{\on}$. By virtue of Lemma \ref{lemma:good_nu}, we have $\|s_{\nu_1}\|_1 \le \|x^\rmR_{\nu_1}\|_1$. Therefore,
\begin{align}
    \|s_{\nu_1,\mathrm{res}}\|_1 &\le \sum_{n\in\caD}\left(|x^R_{\nu_1}| - |s_{\nu_1,n}|\right) + \sum_{n\in\caD_{\res}}|x^\rmR_{\nu_1,n}| \quad \nonumber\\
    &\le \|x^{\mathrm{R}}_{\nu_1} - s_{\nu_1}\|_1 + \pi^2|\nu_1 - \fku|\log N \quad \ \textrm{(Lemma \ref{lemma:bounds})}\nonumber\\
    &\le \sqrt{S}(C_1 + C_2 \pi)\soff + \sqrt{S}\pi^2 \sigma_0 = \caO(\sqrt{S}\sigma_0).
\end{align}
Hence $\|s_{\nu_\ast} - s_{\nu_1}\|_2 = \caO(\sigma_0)$. According to Lemma~\ref{lemma:good_nu}, we have $\|s_{\nu_1} - x_{\mathrm{on}}\|_2 \le C_3 \sigma = \mathcal{O}(\sigma_0)$. Therefore,
\begin{equation}
    \|s_{\nu_\ast} - x_{\mathrm{on}}\|_2 \le \|s_{\nu_\ast} - s_{\nu_1}\|_2 + \|s_{\nu_1} - x_{\on}\|_2 = \caO(\sigma_0).
\end{equation}
Finally, the accuracy on $E_0$ can be confirmed by combining Lemma~\ref{lemma:recover} and the above 2-norm error bound.

\subsection{Auxiliary lemmas}

The following lemmas are critical for the proof of Theorem~\ref{theorem:main}. Henceforth, we assume $s_\nu$ is the solution of
\begin{equation}
    \min_{s\in\mathbb{R}^N} \|s\|_1,\quad \mathrm{s.t.}\quad \|F_{\mathcal{T}}s - y_{\mathcal{T}}\|_2 \le \sqrt{|\mathcal{T}|}\sigma.
\end{equation}
Consider a signal $x = \sum_{n\in\caD'}q_n \delta_n$ that follows the description of Lemma \ref{lemma:recover}. Suppose $p_{\min} \gg S\sigma_{\off}$. Then $\caD$ is the indices corresponding to the $S$ largest entries of $x$, and the remaining part is $\caD_{\res}$, and
\begin{equation}\label{equ:xon}
    x_{\dom} = \sum_{n\in\caD}q_n \delta_n,\quad x_{\res} = \sum_{n\in \caD_{\res}}q_n \delta_n,\quad 
    |\caD| = S,\quad \|x_\res\|_1 = \caO(\soff).
\end{equation}

\begin{lemma} \label{lemma:bounds}
    Suppose $y = Fx$ is a length-$N$ ($N \ge 100$) on-grid signal and $x$ satisfies the condition in Eq.~(\ref{equ:xon}) and $\|x\|_1 \le 1$. Let
    \begin{equation}\label{equ:Cx}
        C[x] \equiv \left[(4 + 2\pi^2/N)\|x\|_2^2 - 2\pi^2\|x\|_1^2/N\right]^{1/2}.
    \end{equation}If $y = F_\nu(x^\rmR_\nu + \rmi x^\rmI_\nu)$, then
    \begin{gather}
        C[x]|\nu| \le \|x^{\mathrm{I}}_\nu\|_2 \le \frac{2\pi}{\sqrt{3}}|\nu|, \label{equ:lemma5fst} \\
        \|x^{\mathrm{R}}_\nu - x\|_2 \le \frac{2\pi}{\sqrt{3}}|\nu|, \label{equ:lemma5snd}\\
        \sum_{n\in\caD_{\res}}|x^{\mathrm{R}}_{\nu,n}| \le \pi^2|\nu|\log N. \label{equ:lemma5trd}
    \end{gather}
\end{lemma}
The proof of Lemma~\ref{lemma:bounds} is presented in Appendix~\ref{app:proof_of_lemma5}, which is a direct corollary of Lemma~\ref{lemma5:first part} and Lemma~\ref{lemma5:second part}.

 In Appendix~\ref{app:proof_of_lemma2}, we prove the following lemma.
\begin{lemma}[A good $\nu$ generates a good solution]
    Suppose $y^0$ is a length-$N$ signal with optimal grid decomposition $y^0 = F_{\mathfrak{u}}(x_{\on}  + \rmi x_\off), \soff = \|F_{\mathfrak{u}}x_\off\|_{\infty}$; the target signal is $y = y^0 + z$ with $\|z\|_{\infty} \le \sH$; let
    
    \begin{equation}\label{equ:C0C3}C_0 = (\soff + \sH)/\sigma,\quad C_3 = 2C_1 + C_2 \pi +\frac{2\pi}{\log N}\sqrt{\frac{S}{3}};\end{equation}
    $\mathcal{T}$ is an integer set generated by sampling from $[N]$ with ratio 
    \begin{equation}
        r = \mathcal{O}\left(S \log^2 (S) \mathrm{poly}\log(N) N^{-1}\right),
    \end{equation}
    and $s_\nu$ is the solution of Eq.~(\ref{equ:compressed_sensing}).
    If
    \begin{equation}\label{equ:lemma5condition}
        |\nu| \le \frac{\sqrt{S}\sigma}{\log N},\quad C_0 \le 1 - \frac{4\pi}{\log N}\sqrt{\frac{S}{3}},
    \end{equation}
    then and with probability at least $1 - 1/\poly(N)$, we have
    \begin{equation} \label{equ:lemma3conclusion}
        \|s_{\nu}\|_1 \le \|x^\rmR_\nu\|_1,\quad \|s_{\nu} - x_{\mathrm{on}}\|_2 \le C_3\sigma.
    \end{equation}
\label{lemma:good_nu}
\end{lemma}

In Appendix~\ref{app:proof_of_lemma3}, we prove the following lemma.
\begin{lemma}[A bad $\nu$ generates a bad solution]
    Suppose $y^0$ is a length-$N$ signal with optimal grid decomposition $y^0 = F_{\mathfrak{u}}(x_{\on}  + \rmi x_\off), \soff = \|F_{\mathfrak{u}}x_\off\|_{\infty}$ and satisfies Eq.~(\ref{equ:assumption}); the target signal is $y = y^0 + z$ with $\|z\|_\infty \le \sH$;
    $\mathcal{T}$ is an integer set generated by sampling from $[N]$ with ratio 
    \begin{equation}
        r = \mathcal{O}\left(S \log^2 (S)  \mathrm{poly}\log(N) N^{-1}\right),
    \end{equation}
    and $s_\nu$ is the solution of Eq.~(\ref{equ:compressed_sensing}); $\mathcal{I}$ is an integer set generated by sampling from $[N]$ with equal probability for $L$ times; let $\sigma_\test$ be a constant with lower bound
    \begin{equation} \label{equ:sigmatestlowerbound}
        \sigma_\test \ge \sqrt{\frac{3}{2}}\left[\left(C_3^2 + \frac{8\pi^2 S}{3\log^2 N}\right)^{1/2} + C_0\right]\sigma.
    \end{equation}
    If $|\nu| \le \sqrt{S}\sigma/\log N$, then with probability at least $(1 - e^{-L/2})(1 - 1/\poly(N))$,
    \begin{equation} \label{equ:lemma6result1}
        \sum_{t\in\mathcal{I}}\left|y_t - \left(F_\nu s_\nu\right)_t\right|^2 \le L\sigma^2_{\mathrm{test}},
    \end{equation}
    If instead
    \begin{equation} \label{equ:lemma6nulowerbound}
        |\nu| > C[x_{\on}]^{-1}(\sqrt{2}\sigma_\test + C_0\sigma),
    \end{equation}
    then with probability at least $(1 - e^{-L/2})(1 - 1/\poly(N))$, we have
        \begin{equation} \label{equ:lemma6result2}
        \sum_{t\in\mathcal{I}}\left|y_t - \left(F_\nu s_\nu\right)_t\right|^2 > L\sigma^2_{\mathrm{test}}.
    \end{equation}
\label{lemma:bad_nu}
\end{lemma}

The list of constants used in the proof are as follows.
\begin{table}[h!]
    \centering
    \caption{Constants used in the proof.}
    \begin{tabular}{ c c }
    \hline
    \hline
       Constant  &  Meaning \\
       \hline
        $C_0$ & $(\soff + \sH)/\sigma$ \\
        $C_1$ & Eq.~(\ref{equ:C1C2}) \\
        $C_2$ & Eq.~(\ref{equ:C1C2}) \\
        $C_3$ & $C_1 + C_2 \pi + 2\pi\sqrt{S}/(\sqrt{3}\log N)$ \\
        $\pi_N$ & $\pi(1-N^{-1})$ \\
    \hline
    \hline
    \end{tabular}
\label{table:constants}
\end{table}

\subsection{Proof of Lemma~\ref{lemma:off_grid}}
\label{app:proof_of_lemma1}

Without loss of generality, we assume $\fku = 0$. After computation, we have
\begin{equation}
    \begin{aligned}
        y_{\off,t} &= \begin{cases}
            -\rmi\sum_{f\in\caF}p_f \sin(\pi \nu_f) e^{\rmi \pi \nu_f}e^{-\rmi 2\pi ft}\quad & t\neq 0\\
            0\quad & t = 0
        \end{cases}\\
        y_{\on,t} &= \begin{cases}
            \sum_{f\in\caF}p_f \cos(\pi \nu_f) e^{\rmi \pi\nu_f}e^{-\rmi 2\pi ft}\quad & t \neq 0\\
            1 & t = 0
        \end{cases}\\
        x_{\on,k} &= \sum_{f\in\caF}p_f \cos[\pi_N(k-Nf)] D_N\left(k-Nf\right),\\
        x_{\off,k} &= \sum_{f\in\caF}p_f \sin[\pi_N (k-N\nu_f)] D_N(k-Nf).
    \end{aligned}
\end{equation}
Let $\caD$ denote the integer part of frequencies in $\caF_\dom$. For each $n\in \caD$, define $f_n$ as the frequency that is closest to it, and
\begin{equation}
    q_n \equiv p_{f_n} \sin(\pi\nu_{f_n})e^{\rmi\pi\nu_{f_n}}.
\end{equation}
Hence,
\begin{equation}
    \sum_{n\in\caD}q_n e^{-\rmi 2\pi ft} = \rmi y_{\off,t} - \sum_{f\in\caF_\res}p_f \sin(\pi\nu_f)e^{\rmi\pi\nu_f}e^{-\rmi 2\pi ft} = \caO(\sigma_\off).
\end{equation}
Consider an approximate Fourier matrix $\tilde{F}$, such that
\begin{equation}
    \tilde{F}_{nk} = \begin{cases}
        e^{-\rmi 2\pi nk/N} & n\not\neq \caD,\\
        e^{-\rmi 2\pi f_nt} & n\in \caD. 
    \end{cases}
\end{equation}
Note that $|Nf_n - n'| \ge 1/2 \ \forall n'\not\in \caD$. Hence, 
\begin{equation}
     |p_f\sin(\pi\nu_f)| = |q_n| = \caO(\sigma_\off)\quad \forall n\in \caD.
\end{equation}

\subsection{\label{app:proofofrecoverlemma}Proof of Lemma~\ref{lemma:recover}}

Suppose the signal has optimal grid decomposition (without loss of generality, we assume $\fku = 0$)
\begin{equation}
    y = F^{-1}(x_{\on} + \rmi x_{\off}),\quad x_{\on},x_{\off}\in \mathbb{R}^N.
\end{equation}
After the computation, we obtain the following.
\begin{gather}
    x_{\on,n} = \sum_{f\in\caF} p_f \cos[\pi_N(n-Nf)]D_N(n-Nf).
\end{gather}
The optimal on-grid condition $\|y_\off\|_{\infty} = \sigma_\off$ requires that $\|x_{\off}\|_2 \le \sigma_\off$. For each $n\in \caD$, define $f_n$ as the frequency that is closest to it, hence
\begin{equation}
    \begin{aligned}
        x_{\on,n} &= \sum_{f\in\caF} p_f \cos[\pi_N(n-Nf)]D_N(n-Nf)\\
        &= p_{f_n}\cos[\pi_N \nu_{f_n}]D_N(\nu_{f_n}) + \sum_{f\in\caF_\res} p_f \cos[\pi_N(n-Nf)]D_N(n-Nf)\\
        &\quad + \sum_{m\neq n\in\caD}p_{f_m} \cos[\pi_N(m-n-\nu_{f_m})]D_N(m-n-\nu_{f_m})
    \end{aligned}
\end{equation}
Note that when $N \ge 100, |x| \ge 1/2$, we have $\pi\cos(\pi_N x)/2x \in [-1.1,0.55]$. Therefore,
\begin{align}
    p_{f_n}\cos[\pi_N \nu_{f_n}]D_N(\nu_{f_n}) &\ge p_{f_n} \left(1 - \frac{1}{2}\pi_N^2 \nu_{f_n}^2\right)\left(1 - \frac{\pi}{3}\nu_{f_n}^2\right),\\
    p_{f_m} \cos[\pi_N(n-Nf)]D_N(n-Nf) &\ge -\pi p_{f_m}|\nu_{f_m}|\left|\frac{\cos[\pi_N(n-Nf)}{N\sin[\pi(n-Nf)/N]}\right| \nonumber \\
    &\ge -1.1|p_{f_m}\nu_{f_m}|,\\
    \sum_{f\in\caF_\res} p_f \cos[\pi_N(n-Nf)]D_N(n-Nf) &\ge -\sum_{f\in\caF_\res}p_f = -\caO(\soff).        
\end{align}
Recall that $\forall f\in\caF_\dom, |p_f\nu_f| = \caO(\sigma_\off)$ (Lemma \ref{lemma:off_grid}). Hence,
\begin{equation}
    x_{\on,n} \ge p_{f_n} - \caO(S\soff).
\end{equation}
Using the same argument, we can show that for all $n\neq \caD$, $x_{\on,n} \le \caO(S\sigma_\off)$.

\subsection{Proof of Lemma~\ref{lemma:good_nu}}
\label{app:proof_of_lemma2}

To prove the lemma, we first need a sufficient condition for $s = x^{\mathrm{R}}_\nu$ to be a feasible solution of Eq.~(\ref{equ:compressed_sensing}), i.e., 
\begin{equation}
    \|F_{{\nu},\mathcal{T}}x^{\mathrm{R}}_\nu - y_\mathcal{T}\|_2 \le \sqrt{|\mathcal{T}|}\sigma,
\end{equation} so that we can bound $\|s_\nu - x^{\mathrm{R}}_\nu\|_2$. Note that
\begin{equation}
    F_{{\nu},\mathcal{T}}x^{\mathrm{R}}_\nu - y_\mathcal{T} = - \rmi F_{\nu,\mathcal{T}}x^{\mathrm{I}}_\nu - y^0_{\mathrm{off},\mathcal{T}} - z_{\mathcal{T}},
\end{equation} 
where $\|y^0_{\mathrm{off},\mathcal{T}}\|_2 \le \sqrt{|\mathcal{T}|}\soff, \|z_{\mathcal{T}}\|_2 \le \sqrt{|\mathcal{T}|}\sH$, and $\|F_{\nu,\mathcal{T}}x^{\mathrm{I}}_\nu\|_2$ has the following upper bound.
\begin{lemma}[Concentration of $\|F_{\nu,\mathcal{T}}x^{\mathrm{I}}_\nu\|_2$]\label{lemma:concentration}
    Suppose $\mathcal{T}$ is an integer set in $[N]$ generated by sampling ratio $r = \mathcal{O}(N^{-1}S^2\poly\log N)$. Then
    \begin{equation}
        \|F_{\nu,\mathcal{T}}x^{\mathrm{I}}_\nu\|_2 \le 4\pi\sqrt{
        \frac{|\mathcal{T}|}{3}}|\nu|
    \end{equation}
    with probability at least $1 - 1/\mathrm{poly}(N)$.
\end{lemma}
\begin{proof}
    Let $V_t \equiv |(F_\nu x^{\mathrm{I}}_\nu)_t|^2$. Given a sampling ratio $r$, we introduce the following random variables:
\begin{equation}
    \mathrm{Pr}[\hat{X}_t = V_t] = r,\quad \mathrm{Pr}[\hat{X}_t = 0] = 1-r;\quad \hat{R} = \sum_{t=0}^{N-1} \hat{X}_t.
\end{equation}
Note that $\hat{R} = \|F_{\nu,\mathcal{T}}x^{\mathrm{I}}_\nu\|_2^2$, whose expectation value is
\begin{equation}
    \mathbb{E}[\hat{R}] = r\|F_\nu x^{\mathrm{I}}_\nu\|_2^2 = Nr \|x^{\mathrm{I}}_\nu\|_2^2.
\end{equation}
Bernstein's inequality states that
\begin{equation}
    \pr\left(\hat{R} \ge 2\mathbb{E}[\hat{R}] \right) \le \exp\left(-\frac{\frac{1}{2}\mathbb{E}[\hat{R}]^2}{\sum_{t=0}^{N-1}E[\hat{X}_t^2] + \frac{1}{3}\max_t V_t \cdot \mathbb{E}[\hat{R}]}\right).
\end{equation}
We need an upper bound for $\sum_{t=0}^{N-1}E[\hat{X}_t^2]$. Note that
\begin{align}
    \frac{1}{r}\sum_{t=0}^{N-1}\mathbb{E}[\hat{X}_t^2] = \sum_{t=0}^{N-1}V_t^2= \sum_{t=0}^{N-1}\left|(F_\nu x^{\mathrm{I}}_\nu)_t\right|^4.
\end{align}
After computation, we obtain
\begin{equation}
\begin{aligned}
    (F_\nu x^{\mathrm{I}}_\nu)_t &= \sum_{k=0}^{N-1}e^{-\rmi2\pi (k+\nu)t/N}x^{\mathrm{I}}_{\nu,k} = e^{-\rmi2\pi\mathfrak{u} t/N}e^{-\rmi\pi(\nu - \fku)}\sin[\pi(\nu - \fku)] y_{\mathrm{on},t}\\
    V_t & = \left|(F_\nu x^{\mathrm{I}}_\nu)_t\right|^2 = \sin^2[\pi(\nu - \fku)]\cdot|y_{\mathrm{on},t}|^2 \le \pi^2|\nu-\mathfrak{u}|^2,\\
    \sum_{t=0}^{N-1}V_t^2 &= \sum_{t=0}^{N-1}\left|(F_\nu x^{\mathrm{I}}_\nu)_t\right|^4 \le \pi^4|\nu-\mathfrak{u}|^4\sum_{t=0}^{N-1}|y_{\mathrm{on,t}}|^4 \le \pi^4 N|\nu-\mathfrak{u}|^4.
\end{aligned}
\end{equation}
In conjunction with Eq.~(\ref{equ:lemma5fst}), we have
\begin{equation}
    \pr\left(\sum_{t=0}^{N-1}\hat{X}_t \ge 2\mathbb{E}[\hat{R}]\right) \le \exp\left[-\frac{-\frac{1}{2}Nr C_4^4}{\pi^4 + \frac{1}{3}\pi^2 C_4^2}\right].
\end{equation}
Since $Nr = \mathcal{O}(S^2\poly\log N)$, we can conclude that
\begin{align}
    \|F_{\nu,\mathcal{T}}x^{\mathrm{I}}_\nu\|_2 &\le \sqrt{\frac{2\mathbb{E}[\mathcal{T}]}{N}}\|F_\nu x^{\mathrm{I}}_\nu\|_2 \le 2\sqrt{\frac{|\mathcal{T}|}{N}}\|F_\nu x^{\mathrm{I}}_\nu\|_2 \nonumber \\
    &\le 2\sqrt{|\mathcal{T}|}\|x^{\mathrm{I}}_{\nu}\|_2 \le 4\pi\sqrt{\frac{|\mathcal{T}|}{3}}|\nu|
\end{align}
with probability at least $1 - 1/\mathrm{poly}(N)$.
\end{proof}

\begin{proof}[Proof of Lemma~\ref{lemma:good_nu}]

The upper bound of $C_0$ in Eq.~(\ref{equ:lemma5condition}) and Lemma~\ref{lemma:concentration} ensure that $x^{\mathrm{R}}_\nu$ is feasible with high probability, which implies $\|s_{\nu}\|_1 \le \|x^{\mathrm{R}}_\nu\|_1$. Meanwhile,
\begin{equation}
    \|F_{{\nu},\mathcal{T}}(s_{\nu} - x^{\mathrm{R}}_\nu)\|_2 \le \|F_{{\nu},\mathcal{T}}s_{\nu} - y_\mathcal{T}\|_2 + \|F_{{\nu},\mathcal{T}} x^{\mathrm{R}}_\nu - y_\mathcal{T}\|_2 \le 2\sqrt{|\mathcal{T}|}\sigma.
\end{equation}
In conjunction with Theorems~\ref{theorem:rip},~\ref{theorem:compressed_sensing}, Lemma~\ref{lemma:bounds} and the condition that $|\nu| \le \sqrt{S}\sigma/\log N$, we obtain
\begin{align}
    \|s_{\nu} - x^{\mathrm{R}}_\nu\|_2 &\le 2C_1 \sigma + C_2 \|x^{\mathrm{R}}_{\nu,\mathrm{res}}\|_1/\sqrt{S}\nonumber \quad (\textrm{Theorems \ref{theorem:rip}, \ref{theorem:compressed_sensing}})\\
    &\le 2C_1 \sigma + C_2 \pi|\nu| \log N/\sqrt{S}\nonumber \quad (\textrm{Lemma \ref{lemma:bounds}})\\
    &\le (2C_1 + C_2 \pi)\sigma,\\
    \|x^{\mathrm{R}}_\nu - x_{\mathrm{on}}\|_2 &\le \frac{2\pi}{\sqrt{3}}|\nu|\nonumber \quad\quad \quad \quad \quad \quad \quad \quad \quad (\textrm{Lemma \ref{lemma:bounds}})\\
    &\le \frac{2\pi}{\sqrt{3}}\frac{\sqrt{S}}{\log N}\sigma.
\end{align}
The lemma is proved by triangle inequality from here.

\end{proof}

\subsection{Proof of Lemma~\ref{lemma:bad_nu}}
\label{app:proof_of_lemma3}

The Chernoff-Hoeffding's inequality states that
\begin{lemma} \label{lemma:Hoeffding}
    Given a complex vector $b\in\mathbb{C}^N$ with $\|b\|_{\infty} \le 1$. Randomly sample integers in $[N]$ for $L$ times and denote the result as $\mathcal{I}$. Then 
    \begin{equation}
    \begin{aligned} \label{equ:lemmaHoeffding}
         \mathrm{Pr}\left(\sum_{i\in \mathcal{I}} |b_{i}|^2 \le \frac{L}{2N}\|b\|_2^2\right) < \exp\left(-\frac{L}{2}\right),\quad 
         \mathrm{Pr}\left(\sum_{i\in\mathcal{I}} |b_i|^2 \ge \frac{3L}{2N}\|b\|_2^2\right) < \exp\left(-\frac{L}{2}\right).
    \end{aligned}
    \end{equation}
\end{lemma}
\begin{proof}
Introduce i.i.d. random variables $\hat{B}_1, \hat{B}_2,\ldots,\hat{B}_L$ satisfying fora all $\ell\in[1,L]$, 
\begin{equation}
    \mathrm{Pr}\left(\hat{B}_l = |b_i|^2\right) = \frac{1}{N},\quad i \in [N].
\end{equation}
The distribution of $\sum_{i\in \mathcal{I}} |b_i|^2$ is identical to that of $\sum_{l=1}^L \hat{B}_l$. The Chernoff-Hoeffding's inequality states that
\begin{equation}
    \mathrm{Pr}\left(\sum_{l=1}^L \hat{B}_l \le \frac{L}{2}\mathbb{E}[\hat{B}]\right) < \exp\left(-\frac{L}{2\max_i|b_i|^2}\right) \le \exp\left(-\frac{L}{2}\right) ,
\end{equation}
where
\begin{equation}
    \mathbb{E}[\hat{B}] = \frac{1}{N}\sum_{i=1}^N |b_i|^2 = \frac{1}{N}\|b\|_2^2.
\end{equation}
Hence, the first concentration inequality in Eq.~(\ref{equ:lemmaHoeffding}) is confirmed. The other one can be proved in the same manner.

\end{proof}

Now we can prove Lemma~\ref{lemma:bad_nu}.
\begin{proof}
Recall that $y = y^0_{\on} + y^0_{\off} + z,\ y^0_\on = F_\nu(x^\rmR_\nu + \rmi x^{\rmI}_\nu)$. Therefore,
\begin{align}
    \|y - F_\nu s_\nu\|_2 &= \|F_\nu(x^\rmR_\nu - s_\nu+ \rmi x^\rmI_\nu) + y^0_{\off} + z\|_2\\
    &\le \|F_\nu(x^\rmR_\nu - s_\nu)\|_2 + \|F_\nu x^{\mathrm{I}}_\nu\|_2 + \|y_\off\|_2 + \|z\|_2\nonumber\\
    &= \sqrt{N}\sqrt{\|x^{\mathrm{R}}_\nu - s_{\nu}\|_2^2 + \|x^{\mathrm{I}}_\nu\|_2^2} + \|y_\off\|_2 + \|z\|_2\nonumber\\
    &\le \sqrt{N}\left(\|x_{\on} - s_\nu\|_2^2 + \|x^\rmR_\nu -  x_{\on}\|_2^2 + \|x^\rmI_\nu\|_2^2\right)^{1/2} + \sqrt{N}\soff + \sqrt{N}\sH\nonumber\\
    &\le \sqrt{N} \left(\|x_{\on} - s_\nu\|_2^2 + \frac{8\pi^2}{3}\nu^2\right)^{1/2} + \sqrt{N}\sigma_\off + \sqrt{N}\sH \quad \textrm{(Lemma~\ref{lemma:bounds})}\nonumber\\
    &\le \sqrt{N}\left(C_3^2\sigma^2 + 8\pi^2\nu^2/3\right)^{1/2} + \sqrt{N}\left(\soff + \sH\right) \quad \textrm{(Lemma~\ref{lemma:good_nu})}.
\end{align}
According to Lemma~\ref{lemma:Hoeffding}, with probability at least $1 - \exp(-L/2)$, we have
\begin{equation}
    \sum_{t\in\mathcal{I}}\left|y_t - \left(F_\nu s_\nu\right)_t\right|^2 \le \frac{3L}{2}\left[\left(C_3^2\sigma^2 + 8\pi^2\nu^2/3\right)^{1/2} + \soff + \sH\right]^2. 
\end{equation}
Equation (\ref{equ:lemma6result1}) is confirmed when we take the upper bound $|\nu| = \sqrt{S}\sigma/\log N$.

 By virtue of Lemma~\ref{lemma:Hoeffding}, with probability at least $1 - \exp(-L/2)$, we have
\begin{equation}
    \sum_{t\in\mathcal{I}}\left|y_t - \left(F_\nu s_\nu\right)_t\right|^2 \ge \frac{L}{2N}\|y - F_\nu s_\nu\|_2^2.
\end{equation}
The 2-norm has the following lower bound.
\begin{align}
    \|y - F_\nu s_\nu\|_2 &\ge \|F_\nu(x^{\mathrm{R}}_\nu + \rmi x^{\mathrm{I}}_\nu - s_{\nu}) + y_{\mathrm{off}} + z\|_2\nonumber\\
    &\ge \|F_\nu(x^{\mathrm{R}}_\nu + \rmi x^{\mathrm{I}}_\nu - s_{\nu})\|_2 - \|y_\off\|_2 - \|z\|_2\nonumber\\
    &\ge \sqrt{N}\|x^\rmI_\nu\|_2 - \sqrt{N}\soff - \sqrt{N}\sH\nonumber\\
    &\ge \sqrt{N}(C[x_{\on}]|\nu| - C_0 \sigma) \quad \quad \quad \textrm{(Lemma~\ref{lemma:bounds})}
\end{align}
Hence, 
\begin{equation}
    \sum_{t\in\mathcal{I}}\left|y_t - \left(F_\nu s_\nu\right)_t\right|^2 \ge \frac{L}{2}(C[x_{\on}]|\nu| - C_0\sigma)^2.
\end{equation}
This completes the proof of Eq.~(\ref{equ:lemma6result2}).
\end{proof}

\section{Proof of technical lemmas}

\subsection{Properties of the Dirichlet kernel}

The Dirichlet kernel is defined as
\begin{equation}
    D_N(\nu) \equiv \frac{1}{N}\sum_{m=0}^{N-1}e_{\nu}^{-N+1+2m},\quad e_{\nu} \equiv e^{\rmi\pi \nu/N}.
\end{equation}
In a more concise form, it equals
\begin{equation}
    D_N(\nu) = \begin{cases}
        1,\quad & \nu = 0\\
        \frac{\sin(\pi\nu)}{N\sin(\pi\nu/N)},\quad & \nu \neq 0.
    \end{cases}
\end{equation}
We start with a few estimations for the Dirichlet kernel.
\begin{lemma} \label{lemma:Dirichlet}
    The Dirichlet kernel satisfies
\begin{gather}
    1 - D_N(\nu)^2 \le \frac{\pi\nu^2}{3},\label{equ:dirichleteq1}\\
    |D_N(n + \nu)| \le \frac{\pi|\nu|}{2|n+\nu|},\label{equ:dirichleteq2}\\
    \sum_{n=0}^{N-1}D_N(n + \nu)D_N(n + \nu + l) = \delta_{l,0}.\label{equ:dirichleteq3}
\end{gather}
\end{lemma}

\begin{proof} If $1 - D_N(\nu) \le c\nu^2$, then
\begin{equation}
    1 - D_N^2(\nu) \le 1 - (1-c\nu^2)^2 \le 2c\nu^2.
\end{equation}
Therefore, we work on $1 - D_N(\nu)$ instead. Suppose $N$ is odd. Then
\begin{align}
        1 - D_N(\nu) &= \frac{1}{N}\sum_{m=0}^{N-1}(1-e_{\nu}^{-N+1+2m}),\nonumber \\
        &= \frac{1}{N}\sum_{m'=1}^{\frac{N-1}{2}}(2 - e_{\nu}^{2m'} - e_{\nu}^{-2m'})= \frac{1}{N}\sum_{m'=1}^{\frac{N-1}{2}}4\sin^2\left(\frac{m'\pi\nu}{N}\right)\nonumber \\
        &\le \frac{4}{N}\sum_{m'=1}^{\frac{N-1}{2}}\frac{\pi^2 (m')^2 \nu^2}{N^2}= \frac{4\pi^2}{N^3}\frac{1}{6}\frac{N-1}{2}\frac{N+1}{2}N\nu^2 \le \frac{\pi^2}{6}\nu^2.
\end{align}
Hence, $1 - D_N(\nu)^2 \le \frac{\pi^2}{3}\nu^2$. The even $N$ situation is similar. Equation (\ref{equ:dirichleteq1}) is confirmed.

Equation (\ref{equ:dirichleteq2}) is confirmed by
\begin{equation}
    |D(n+\nu)| = \left|\frac{\sin[\pi(n+\nu)]}{N\sin[\pi(n+\nu)/N]}\right| \le \frac{\pi|\nu|}{2|n+\nu|}.
\end{equation}

If $l = 0$, then
\begin{align}
    D_N(n+\nu) &= \frac{1}{N} \sum_{m=0}^{N-1}e_{n+\nu}^{-N+1+2m},\quad e_{n+\nu} \equiv e^{\rmi\pi(n+\nu)/N},\\
    D_N(n+\nu)^2 &= \frac{1}{N^2}\sum_{m,m'=0}^{N-1}e_{n+\nu}^{-2N+2+2m + 2m'}\nonumber\\
    &= \frac{1}{N} + \frac{1}{N^2}\sum_{m+m' \neq N-1}e_{n+\nu}^{-2N+2+2m+2m'}.
\end{align}
Note that for any $k \neq 0$,
\begin{equation}
    \sum_{n=0}^{N-1}e_{n+\nu}^k  = \sum_{n=0}^{N-1}e^{\rmi\pi(n+\nu)k/N} = e^{\rmi k\nu\pi/N}\sum_{n=0}^{N-1}e^{\rmi\pi nk/N} = 0.
\end{equation}
Hence,
\begin{equation}
    \sum_{n=0}^{N-1}D_N(n+\nu)^2 = 1.
\end{equation}
If $l \neq 0$, then
\begin{align}
    \sum_{n=0}^{N-1}D_N(n+\nu)D_N(n+\nu+l) &= \frac{1}{N^2}\sum_{m_1, m_2 = 0}^{N-1}N\delta_{m_1 + m_2 + 1 - N} \exp\left[\rmi\frac{\pi l}{N}(2m_2 - N + 1)\right]\nonumber\\
    &=\frac{1}{N}\sum_{m=0}^{N-1}\exp\left[\rmi\frac{\pi l}{N}(2m - N + 1)\right] = 0.
\end{align}
Equation (\ref{equ:dirichleteq3}) is confirmed.

\end{proof}

\subsection{Proof of Lemma~\ref{lemma:bounds}}
\label{app:proof_of_lemma5}

Define vectors $\boldsymbol{c}_\nu, \boldsymbol{s}_\nu \in \mathbb{R}^N$ by
\begin{equation} 
\boldsymbol{c}_{\nu,k} \equiv \cos[\pi_N(k+\nu)]D_N(k+\nu), \quad
\boldsymbol{s}_{\nu,k} \equiv \sin[\pi_N(k+\nu)]D_N(k+\nu).
\end{equation}

\begin{lemma}\label{lemma:svcv} Suppose $N \ge 100$. Then
    \begin{gather}
        \|\boldsymbol{s}_{\nu}\|_1 \le |\boldsymbol{s}_{\nu,0}| + \pi^2|\nu|\log N, \label{equation:sv1}\\
        \|\boldsymbol{c}_{\nu}\|_1 \le |\boldsymbol{c}_{\nu,0}| + \pi^2|\nu|\log N, \label{equation:cv1}\\
        2|\nu| \le \|\boldsymbol{s}_\nu\|_2 \le \frac{2\pi}{\sqrt{3}}|\nu|,\label{equation:sv2}\\ 
        \|\boldsymbol{c}_\nu - \delta_0\|_2 \le \frac{2\pi}{\sqrt{3}}|\nu|.\label{equation:cv2}
    \end{gather}
\end{lemma}

\begin{proof}
By virtue of Lemma~\ref{lemma:Dirichlet}, we have
\begin{gather}
    |\boldsymbol{s}_{\nu,0}| \le \pi|\nu|;\quad |\boldsymbol{s}_{\nu,k}| \le \frac{\pi|\nu|}{2|k + \nu|}, k \neq 0.
\end{gather}
Suppose $N$ is odd. Then
\begin{align}
    \sum_{k=0}^{N-1}|\boldsymbol{s}_{\nu,k}|^2 &\le |\sin\left(\pi_N\nu\right)|^2  D_N(\nu)^2 + \sum_{k\neq 0}D_N(k+\nu)^2\nonumber\\
    &\le 1 - D_N(\nu)^2 + |\sin\left(\pi_N\nu\right)|^2 \le \frac{4}{3}\pi^2\nu^2,\\
    \sum_{k\neq 0}|\boldsymbol{s}_{\nu,k}| &\le \sum_{k\neq 0}|D_N(k+\nu)| \le \frac{\pi|\nu|}{2}\sum_{k=1}^{\frac{N-1}{2}}\left(\frac{1}{k+\nu} + \frac{1}{k-\nu}\right)\nonumber\\
    &\le \pi|\nu|\sum_{k=1}^{\frac{N-1}{2}}\left(\frac{1}{2k-1} + \frac{1}{2k}\right) \le \pi|\nu|\left(1 + \ln\frac{N-1}{2}\right).
\end{align}
One can verify that when $N \ge 100$,
\begin{equation}
    1 + \ln\frac{N-1}{2} \le \pi\log N.
\end{equation}
The proof of even $N$ situation is similar. This completes the proof of Eq.~(\ref{equation:sv1}) and the upper bound of Eq.~(\ref{equation:sv2}). The proof for Eq.~(\ref{equation:cv1}) is similar.

Note that
\begin{equation}
\begin{aligned}
    \sin^2[\pi(k+\nu)] - \sin^2[\pi_N(k+\nu)] &= \frac{1}{2}\left[\cos(2\pi\nu) - \cos(2\pi_N(k+\nu))\right]\\
    = \cos(2\pi\nu)\sin^2\left[\frac{\pi(k+\nu)}{N}\right] &+ \sin(2\pi\nu)\sin\left[\frac{\pi(k+\nu)}{N}\right]\cos\left[\frac{\pi(k+\nu)}{N}\right].
\end{aligned}
\end{equation}
Therefore,
\begin{align}
    \|\boldsymbol{s}_\nu\|_2^2
    &= \sin^2(\pi\nu) + \cos(2\pi\nu)\sum_{k=0}^{N-1}\sin^2\left[\frac{\pi(k+\nu)}{N}\right]D_N(k+\nu)^2 \nonumber\\
    &\quad - \sin(2\pi\nu)\sum_{k=0}^{N-1}\sin\left[\frac{\pi(k+\nu)}{N}\right]\cos\left[\frac{\pi(k+\nu)}{N}\right]D_N(k+\nu)^2.
\end{align}
Let $e_{k+\nu} \equiv e^{\rmi \pi(k+\nu)/N}$. After computation, we obtain the following.
\begin{align}
    \sum_{k=0}^{N-1}\sin^2\left[\frac{\pi(k+\nu)}{N}\right]D_N(k+\nu)^2 &= \sum_{k=0}^{N-1}\frac{2 - e_{k+\nu}^2 - e_{k+\nu}^{-2}}{4N^2}\sum_{m,m'=0}^{N-1} e_{k+\nu}^{-2N + 2 + 2m + 2m'} \nonumber\\
    &= \frac{1}{2} - \frac{1}{4N}(N-1+e^{\rmi 2\pi\nu}) - \frac{1}{4N}(N-1+e^{-\rmi 2\pi\nu}) \nonumber\\
    &= \frac{1}{N}\sin^2(\pi\nu),
\end{align}
and
\begin{align}
 &\quad \sum_{k=0}^{N-1}\sin\left[\frac{\pi(k+\nu)}{N}\right]\cos\left[\frac{\pi(k+\nu)}{N}\right]D_N(k+\nu)^2\nonumber\\
    &= \sum_{k=0}^{N-1}\frac{e_{k+\nu}^2 - e_{k+\nu}^{-2}}{4N^2\rmi}\sum_{m,m'=0}^{N-1} e_{k+\nu}^{-2N + 2 + 2m + 2m'}\nonumber\\
    &= \frac{1}{4N\rmi}(N-1+e^{\rmi 2\pi\nu}) - \frac{1}{4N\rmi}(N-1+e^{-\rmi 2\pi\nu}) = \frac{1}{2N}\sin(2\pi\nu).
\end{align}
Hence, 
\begin{equation}
\begin{aligned}
    \|s_\nu\|_2^2 &= \sin^2(\pi\nu) + \frac{1}{N}\cos(2\pi\nu)\sin^2(\pi\nu) - \frac{1}{2N}\sin^2(2\pi\nu)\\
    &= (1 - 2N^{-1})\sin^2(\pi\nu).
\end{aligned}
\end{equation}
In conjunction with the condition $N \ge 100$, we obtain the lower bound in Eq.~(\ref{equation:sv2}).

Finally,
\begin{align}
    \|\boldsymbol{c}_\nu - \delta_0\|_2^2 &\le |1-\boldsymbol{c}_{\nu,0}|^2 + \sum_{k\neq 0}|\boldsymbol{c}_{\nu,k}|^2\nonumber\\
    &\le |1 - \cos\left(\pi_N\nu\right)D_N(\nu)|^2 + 1 - |D_N(\nu)|^2\nonumber\\
    &\le 2 - 2|D_N(\nu)| + \sin^2\left(\pi_N\nu\right) D_N(\nu)^2 \le \frac{4}{3}\pi^2|\nu|^2.
\end{align}
This completes the proof of Eq.~(\ref{equation:cv2}).

\end{proof}

The first part of Lemma~\ref{lemma:bounds} is proved in the following lemma.
\begin{lemma} \label{lemma5:first part}
Suppose $y = Fx$ is a length-$N$ ($N \ge 100$) on-grid signal and $x$ satisfies the condition in Eq.~(\ref{equ:xon}). Let $y = F_\nu(x^\rmR_\nu + \rmi x^\rmI_\nu)$. Then
    \begin{gather}
        \frac{\|x^{\mathrm{I}}_\nu\|_2}{\|x\|_1} \le \frac{2\pi}{\sqrt{3}}|\nu|,\quad \frac{\|x^{\mathrm{R}}_\nu - x\|_2}{\|x\|_1} \le \frac{2\pi}{\sqrt{3}}|\nu|,\\
        \sum_{n\in\caD_{\res}}|x^\rmR_{\nu,n}| \le \|x\|_1\pi^2|\nu|\log N.
    \end{gather}
\end{lemma}
\begin{proof}
By definition, we have
\begin{align}
    x^{\mathrm{R}}_\nu = \sum_{n\in\mathcal{D'}} q_n \boldsymbol{c}^{(n)}_\nu,\quad \boldsymbol{c}^{(n)}_{\nu,k} \equiv \cos\left[\pi_N(k-n+\nu)\right]D_N(k-n+\nu),\\
    x^{\mathrm{I}}_\nu = \sum_{n\in\mathcal{D'}} q_n \boldsymbol{s}^{(n)}_\nu,\quad \boldsymbol{s}^{(n)}_{\nu,k} \equiv \sin\left[\pi_N(k-n+\nu)\right]D_N(k-n+\nu).
\end{align}
Note that $\boldsymbol{s}^{(n)}_\nu, \boldsymbol{c}^{(n)}_\nu$ are simply $\boldsymbol{s}_{\nu}, \boldsymbol{c}_{\nu}$ with a permutation in entries.
In Lemma~\ref{lemma:svcv}, we have proved that
\begin{equation}
    \|\boldsymbol{s}_{\nu}\|_2 \le \frac{2\pi}{\sqrt{3}}|\nu|,\quad \|\boldsymbol{c}_{\nu} - \delta_n\|_2 \le \frac{2\pi}{\sqrt{3}}|\nu|.
\end{equation}
Hence,
\begin{align}
    \|x^{\mathrm{R}}_\nu - x\|_2 &\le \sum_{n\in\mathcal{D'}} q_n\|\boldsymbol{c}^{(n)}_\nu - \delta_n\|_2 \le \|x\|_1\frac{2\pi}{\sqrt{3}}|\nu|,\\
    \|x^{\mathrm{I}}_\nu\|_2 &\le \sum_{n\in\mathcal{D'}} q_n \|\boldsymbol{s}^{(n)}_\nu\|_2 \le \|x\|_1\frac{2\pi}{\sqrt{3}}|\nu|.
\end{align}
Similarly, in conjunction with the first two inequalities in Lemma~\ref{lemma:svcv}, we obtain
\begin{align}
    \sum_{n\in\caD_{\res}}|x^\rmR_{\nu,n}| &\le \sum_{m\in\mathcal{D}'} q_m \left(\sum_{n\in\caD_\res}\left|\boldsymbol{c}^{(m)}_{\nu,n}\right|\right) \le \|x\|_1\pi^2|\nu|\log N .
\end{align}

\end{proof}

The most critical part is the lower bound of $\|x^{\mathrm{I}}_\nu\|_2$. We prove it separately in the following lemma.
\begin{lemma} \label{lemma5:second part}

Suppose $y = Fx$ is a length-$N$ ($N \ge 100$) on-grid signal and $x$ satisfies the condition in Eq.~(\ref{equ:xon}). Let $y = F_\nu(x^\rmR_\nu + \rmi x^\rmI_\nu)$. Then $\|x^{\mathrm{I}}_\nu\|_2 \ge C[x] |\nu|$ with function $C$ defined in Eq.~(\ref{equ:Cx}).    

\label{lemma:lower bound}
\end{lemma}
\begin{proof}
Denote the on-grid component as $x = \sum_{n\in\caD'}q_n \delta_n$, then
\begin{equation}
    \|x_{\nu}^\rmI\|_2^2 = \sum_{n\in\mathcal{D'}}q_n^2\mathcal{M}_{n,n} + \sum_{n\neq m\in\mathcal{D'}}q_n q_m\mathcal{M}_{n,m},
\end{equation}
where
\begin{equation}
    \mathcal{M}_{n,m} \equiv \sum_{k=0}^{N-1}\sin\left[\pi_N(k-n+\nu)\right]\sin\left[\pi_N (k-m+\nu)\right]D_N(k-n+\nu)D_N(k-m+\nu).
\end{equation}
By virtue of Lemma~\ref{lemma:svcv}, for all $n$, we have
\begin{equation}
    \mathcal{M}_{n,n} = \sum_{k=0}^{N-1}\sin^2\left[\pi_N(k-n+\nu)\right]D^2_N(k-n+\nu) = \|\boldsymbol{s}_{\nu}\|_2^2 \ge 4\nu^2.
\end{equation}
By virtue of Lemma~\ref{lemma:Dirichlet}, for all $n \neq m$, we have
\begin{equation}
\begin{split}
    \mathcal{M}_{n,m}
    &= \sum_{k=0}^{N-1}\cos[\pi_N(n - m)]D_N(k-n+\nu)D_N(k-m+\nu)\\
    &\quad - \sum_{k=0}^{N-1}\cos\left[\pi_N(2k - n - m + 2\nu)\right]D_N(k-n+\nu)D_N(k-m+\nu)\\
    &= - \sum_{k=0}^{N-1}\cos\left[\pi_N(2k - n - m + 2\nu)\right]D_N(k-n+\nu)D_N(k-m+\nu).
\end{split}
\end{equation}
Let $e_{k-n+\nu} \equiv e^{\rmi\pi(k - n + \nu)/N}, e_{k-m+\nu} \equiv e^{\rmi\pi(k-m+\nu)/N}$. Then
\begin{align}
    -\mathcal{M}_{n,m} &= \sum_{k=0}^{N-1}\frac{1}{2}(e_{k-n+\nu}^{N-1} e_{k-m+\nu}^{N-1} + e_{k-n+\nu}^{1-N}e_{k-m+\nu}^{1-N}) \cdot \frac{1}{N^2}\sum_{j=0}^{N-1}\sum_{j'=0}^{N-1}e_{k-n+\nu}^{-N+1+2j}e_{k-m+\nu}^{-N+1+2j'}\nonumber\\
    &= \frac{1}{2N^2}\sum_{j=0}^{N-1}\sum_{j'=0}^{N-1}\sum_{k=0}^{N-1}e_{k-n+\nu}^{2j}e_{k-m+\nu}^{2j'} + \frac{1}{2N^2}\sum_{j=0}^{N-1}\sum_{j'=0}^{N-1}\sum_{k=0}^{N-1}e_{k-n+\nu}^{-2N+2+2j}e_{k-m+\nu}^{-2N+2+2j'}\nonumber\\
    &= \frac{1}{2N^2}\sum_{j=0}^{N-1}\sum_{j'=0}^{N-1}\sum_{k=0}^{N-1}e_{k-n+\nu}^{2j}e_{k-m+\nu}^{2j'} + \frac{1}{2N^2}\sum_{j=0}^{N-1}\sum_{j'=0}^{N-1}\sum_{k=0}^{N-1}e_{n-k-\nu}^{2j}e_{m-k-\nu}^{2j'}\nonumber\\
    &:= M_1 + M_2.
\end{align}
Note that
\begin{align}
    \frac{1}{N}\sum_{k=0}^{N-1}
e_{k-n+\nu}^{2j}e_{k-m+\nu}^{2j'} &= \begin{cases}
    \delta_{j' = 0} & j = 0\\
    \delta_{j' = N-j}e^{\rmi 2\pi j(m-n)/N}e^{-\rmi 2\pi(m-\nu)} & j\neq 0
\end{cases};\\
    \frac{1}{N}\sum_{k=0}^{N-1}e_{n-k-\nu}^{2j}e_{m-k-\nu}^{2j'} &= \begin{cases}
    \delta_{j' = 0} & j = 0\\
    \delta_{j' = N-j} e^{\rmi 2\pi j(n-m)/N} e^{\rmi 2\pi (m-\nu)} & j\neq 0
\end{cases}.
\end{align}
Therefore,
\begin{align}
    M_1 &= \frac{1}{2N}\left[1 + e^{-\rmi 2\pi(m-\nu)}\sum_{j=1}^{N-1}e^{\rmi 2\pi j(m-n)/N}\right]\nonumber\\
     &= \frac{1}{2N}\left[1 - e^{-\rmi 2\pi(m-\nu)}\right] = \frac{1}{2N}(1 - e^{\rmi 2\pi\nu}).\\
 M_2 &= \frac{1}{2N}\left[1 + e^{\rmi 2\pi (m-\nu)}\sum_{j=1}^{N-1}e^{\rmi 2\pi j(n-m)/N}\right]\nonumber\\
     &= \frac{1}{2N}\left[1 - e^{\rmi 2\pi (m-\nu)}\right] = \frac{1}{2N}(1 - e^{-\rmi 2\pi \nu}).
\end{align}
Hence, $\mathcal{M}_{m,n} = (e^{\rmi 2\pi\nu} + e^{-\rmi 2\pi\nu} - 2)/(2N) = -2\sin^2(\pi\nu)/N$, and
\begin{equation}
\begin{aligned}
    \|x^\rmI_\nu\|_2^2 &\ge 4\nu^2\|x\|_2^2 - 2(\|x\|_1^2 - \|x\|_2^2)\sin^2(\pi\nu)/N\\
    &\ge \nu^2\left[(4 + 2\pi^2/N)\|x\|_2^2 - 2\pi^2\|x\|_1^2/N\right].
\end{aligned}
\end{equation}

\end{proof}

\end{document}